\documentclass[11pt]{article}
\usepackage[margin=1in]{geometry}
\usepackage{complexity}
\usepackage{mathtools}
\usepackage{xspace}
\usepackage{rpmacros}
\RequirePackage[colorlinks=true]{hyperref}
\hypersetup{
  linkcolor=[rgb]{0,0,0.4},
  citecolor=[rgb]{0, 0.4, 0},
  urlcolor=[rgb]{0.6, 0, 0}
}
\usepackage{mathpazo}
\usepackage{bbm}
\usepackage{dsfont}

\usepackage{amsthm}
\usepackage{thmtools,thm-restate}

\numberwithin{equation}{section}
\declaretheoremstyle[bodyfont=\it,qed=\qedsymbol]{noproofstyle}

\declaretheorem[name=Observation,numbered=no]{observation*}

\declaretheorem[numberlike=equation]{fact}

\declaretheorem[numberlike=equation]{theorem}
\declaretheorem[numberlike=equation,style=noproofstyle,name=Theorem]{theoremwp}
\declaretheorem[name=Theorem,numbered=no]{theorem*}

\declaretheorem[numberlike=equation]{lemma}
\declaretheorem[name=Lemma,numbered=no]{lemma*}

\declaretheorem[numberlike=equation]{corollary}
\declaretheorem[name=Corollary,numbered=no]{corollary*}
\declaretheorem[numberlike=equation,style=noproofstyle,name=Corollary]{corollarywp}

\declaretheorem[numberlike=equation]{proposition}
\declaretheorem[name=Proposition,numbered=no]{proposition*}
\declaretheorem[numberlike=equation,style=noproofstyle,name=Proposition]{propositionwp}

\declaretheorem[numberlike=equation]{claim}
\declaretheorem[name=Claim,numbered=no]{claim*}

\declaretheorem[name=Conjecture,numbered=no]{conjecture*}

\declaretheorem[numberlike=equation]{question}
\declaretheorem[name=Question,numbered=no]{question*}

\declaretheoremstyle[bodyfont=\it,qed=$\lozenge$]{defstyle} 

\declaretheorem[numberlike=equation,style=defstyle]{definition}
\declaretheorem[unnumbered,name=Definition,style=defstyle]{definition*}

\declaretheorem[unnumbered,name=Example,style=defstyle]{example*}

\declaretheorem[unnumbered,name=Notation=defstyle]{notation*}

\declaretheorem[numberlike=equation,style=defstyle]{construction}
\declaretheorem[unnumbered,name=Construction,style=defstyle]{construction*}

\declaretheorem[unnumbered,name=Remark,style=defstyle]{remark*}

\usepackage{nth}
\usepackage{intcalc}
\usepackage{etoolbox}
\usepackage{xstring}
\hypersetup{
}

\usepackage{ifpdf}
\ifpdf
\else
\usepackage[quadpoints=false]{hypdvips}
\fi

\newcommand{\ECCCurl}[2]{http://eccc.weizmann.ac.il/report/\ifnumcomp{#1}{>}{93}{19}{20}#1/#2/}

\newcommand{\shortECCC}[2]{\texttt{\href{http://eccc.weizmann.ac.il/report/\ifnumcomp{#1}{>}{93}{19}{20}#1/#2/}{eccc:TR#1-#2}}}

\newcommand{\parseECCC}[1]{
\StrSubstitute{#1}{TR}{}[\tmpstring]%
\IfSubStr{\tmpstring}{/}{ 
\StrBefore{\tmpstring}{/}[\ecccyear]%
\StrBehind{\tmpstring}{/}[\ecccreport]%
}{
\StrBefore{\tmpstring}{-}[\ecccyear]%
\StrBehind{\tmpstring}{-}[\ecccreport]%
}%
\shortECCC{\ecccyear}{\ecccreport}}

\newcommand{\coeff}{\mathbf{coeff}}

\newcommand{\vecalpha}{\boldsymbol{\alpha}}
\newcommand{\vecbeta}{\boldsymbol{\beta}}

\newcommand{\Prc}{P^{\mathrm{RC}}}
\newcommand{\Grc}{\cG^{\mathrm{RC}}}

\newcommand{\SPSk}{\ensuremath{\Sigma^k\Pi\Sigma}\xspace}
\declaretheorem[numberlike=equation,name=Meta-Conjecture]{metaconjecture}
\newcommand{\demph}[1]{\textbf{#1}}
\newcommand{\bits}{\{0,1\}}
\newcommand{\rk}{\mathsf{k}}
\newcommand{\rf}{\mathsf{f}}
\newcommand{\RR}{Razborov and Rudich~\cite{RR97}\xspace}
\newcommand\Burgisser{B{\"u}rgisser}
\DeclareMathOperator{\ideg}{ideg}
\DeclareMathOperator{\homog}{hom}
\newcommand{\eqdef}{\vcentcolon=}
\newcommand{\quasi}{\lang{quasi}}

\newcommand{\totdeg}{\ensuremath{\F[x_1,\ldots,x_n]^d}\xspace}
\newcommand{\totdeghom}{\ensuremath{\F[x_1,\ldots,x_n]^{d}_{\homog}}\xspace}
\newcommand{\inddeg}{\ensuremath{\F[x_1,\ldots,x_n]^{d}_{\ideg}}\xspace}
\newcommand{\inddegm}{\ensuremath{\F[x_1,\ldots,x_n]^{1}_{\ideg}}\xspace}

\newcommand{\totdegs}{\ensuremath{\F[\vx]^d}\xspace}
\newcommand{\totdeghoms}{\ensuremath{\F[\vx]^{d}_{\homog}}\xspace}
\newcommand{\inddegs}{\ensuremath{\F[\vx]^{d}_{\ideg}}\xspace}
\newcommand{\inddegsm}{\ensuremath{\F[\vx]^{1}_{\ideg}}\xspace}

\newcommand{\totdegN}{\ensuremath{{N_{n,d}}}\xspace}
\newcommand{\totdeghomN}{\ensuremath{N_{n,d}^{\homog}}\xspace}
\newcommand{\inddegN}{\ensuremath{N_{n,d}^{\ideg}}\xspace}
\newcommand{\inddegNm}{\ensuremath{N_{n,1}^{\ideg}}\xspace}

\newcommand{\coefftot}{\ensuremath{\coeff_{n,d}}\xspace}
\newcommand{\coefftothom}{\ensuremath{\coeff_{n,d}^{\homog}}\xspace}
\newcommand{\coeffind}{\ensuremath{\coeff_{n,d}^{\ideg}}\xspace}

	\renewcommand{\vec}[1]{{\mathbf{#1}}}

	\makeatletter
	\newcommand{\va}{{\vec{a}}\@ifnextchar{^}{\!\:}{}}
	\newcommand{\vb}{{\vec{b}}\@ifnextchar{^}{\!\:}{}}
	\newcommand{\vc}{{\vec{c}}\@ifnextchar{^}{\!\:}{}}
	\newcommand{\vd}{{\vec{d}}\@ifnextchar{^}{\!\:}{}}
	\newcommand{\ve}{{\vec{e}}\@ifnextchar{^}{\!\:}{}}
	\newcommand{\vy}{{\vec{y}}\@ifnextchar{^}{\!\:}{}}
	\newcommand{\vs}{{\vec{s}}\@ifnextchar{^}{\!\:}{}}
	\newcommand{\vt}{{\vec{t}}\@ifnextchar{^}{\!\:}{}}
	\newcommand{\vx}{{\vec{x}}\@ifnextchar{^}{}{}}		
	\newcommand{\vz}{{\vec{z}}\@ifnextchar{^}{\!\:}{}}

	\newcommand{\vY}{{\vec{Y}}\@ifnextchar{^}{\!\:}{}}
	\newcommand{\vX}{{\vec{X}}\@ifnextchar{^}{}{}}		
	\newcommand{\vZ}{{\vec{Z}}\@ifnextchar{^}{\!\:}{}}
	\newcommand{\vG}{{\vec{G}}\@ifnextchar{^}{\!\:}{}}
	\newcommand{\vecX}{\vX}
	\makeatother

	\newcommand{\vaa}{{\vecalpha}}
	\newcommand{\vbb}{{\vecbeta}}
	
\newcommand{\cC}{{\mathcal{C}}}
\newcommand{\cD}{{\mathcal{D}}}
\newcommand{\cE}{{\mathcal{E}}}
\newcommand{\cH}{{\mathcal{H}}}
\newcommand{\cG}{{\mathcal{G}}}

\newcommand{\cM}{{\mathcal{M}}}

\newcommand{\cvG}{{\boldsymbol{\mathcal{G}}}}

\newcommand{\SMESP}[1]{\Sigma \mathrm{m} \! \wedge \! \Sigma\Pi^{#1}}

\newcommand{\trdeg}[1]{\mathsf{trdeg}\left\{#1\right\}}
\newcommand{\Char}{\mathsf{char}}

\newcommand{\powerproduct}{\times\mbox{\small$\wedge$}} 

\newcommand{\Jac}{\mathcal{J}} 

\newcommand{\SSV}{\mathrm{SSV}}
\newcommand{\SSSV}{\mathrm{SSSV}}
\newcommand{\ASSS}{\mathrm{ASSS}}
\newcommand{\FS}{\mathrm{FS}}
\newcommand{\BMS}{\mathrm{BMS}}

\def\epsilon{\varepsilon} 
\let\eps\epsilon

\newcommand*\samethanks[1][\value{footnote}]{\footnotemark[#1]}

\date{}

\title{Succinct Hitting Sets and Barriers to Proving Algebraic Circuits Lower Bounds}
\author{
Michael A. Forbes\thanks{%
	University of Illinois at Urbana-Champaign.
	E-mail: \texttt{miforbes@illinois.edu}.
	This work was performed when the author was at Stanford University, while supported by the NSF, including NSF CCF-1617580, and the DARPA Safeware program.}%
\and%
Amir Shpilka\thanks{Department of Computer Science, Tel Aviv University, Tel Aviv, Israel, E-mails: \texttt{shpilka@post.tau.ac.il, benleevolk@gmail.com}. The research leading to these results has received funding from the European Community's Seventh Framework Programme (FP7/2007-2013) under grant agreement number 257575  and from the Israel Science Foundation (grant number 552/16).}
\and%
Ben Lee Volk\samethanks[2]
}
\begin{document}
\maketitle

\begin{abstract}

	We formalize a framework of \emph{algebraically natural} lower bounds for algebraic circuits.  Just as with the natural proofs notion of Razborov and Rudich~\cite{RR97} for boolean circuit lower bounds, our notion of algebraically natural lower bounds captures nearly all lower bound techniques known.  However, unlike the boolean setting, there has been no concrete evidence demonstrating that this is a \emph{barrier} to obtaining super-polynomial lower bounds for general algebraic circuits, as there is little understanding whether algebraic circuits are expressive enough to support ``cryptography'' secure against algebraic circuits.

	Following a similar result of Williams~\cite{Williams16} in the boolean setting, we show that the existence of an algebraic natural proofs barrier is \emph{equivalent} to the existence of \emph{succinct} derandomization of the polynomial identity testing problem.  That is, whether the coefficient vectors of $\polylog(N)$-degree $\polylog(N)$-size circuits is a hitting set for the class of $\poly(N)$-degree $\poly(N)$-size circuits.  Further, we give an explicit universal construction showing that \emph{if} such a succinct hitting set exists, then our universal construction suffices.

	Further, we assess the existing literature constructing hitting sets for restricted classes of algebraic circuits and observe that \emph{none} of them are succinct as given.  Yet, we show how to modify some of these constructions to obtain succinct hitting sets.  This constitutes the first evidence supporting the existence of an algebraic natural proofs barrier.

	Our framework is similar to the Geometric Complexity Theory (GCT) program of Mulmuley and Sohoni~\cite{GCT1}, except that here we emphasize constructiveness of the proofs while the GCT program emphasizes symmetry.  Nevertheless, our succinct hitting sets have relevance to the GCT program as they imply lower bounds for the complexity of the defining equations of polynomials computed by small circuits.
\end{abstract}

\thispagestyle{empty}
\newpage
\pagenumbering{arabic}

\section{Introduction}\label{sec:intro}

Computational complexity theory studies the limits of efficient computation, and a particular goal is to quantify the power of different computational resources such as time, space, non-determinism, and randomness.  Such questions can be instantiated as asking to prove equalities or separations between complexity classes, such as resolving $\P$ versus $\NP$.  Indeed, there have been various successes: the (deterministic) time-hierarchy theorem showing that $\P\ne\EXP$ (\cite{HS65}), circuit lower bounds showing that $\AC^0\ne \P$ (\cite{Ajtai83,FurstSS84,Yao85,Hastad89}), and interactive proofs showing $\IP=\PSPACE$ (\cite{LFKN92,Shamir90}). However, for each of these seminal works we have now established \emph{barriers} for why their underlying techniques \emph{cannot} resolve questions such as $\P$ versus $\NP$. Respectively, the above results are covered by the barriers of relativization of Baker, Gill and Solovay~\cite{BGS75}, natural proofs of Razborov and Rudich~\cite{RR97}, and algebraization of Aaronson and Wigderson~\cite{AW09}. In this work we revisit the natural proofs barrier of Razborov and Rudich~\cite{RR97} and seek to understand how it extends to a barrier to algebraic circuit lower bounds.  While previous works have considered versions of an algebraic natural proofs barrier, we give the \emph{first} evidence of such a barrier against restricted algebraic reasoning.

\paragraph{Natural Proofs:} The setting of Razborov and Rudich~\cite{RR97} is that of \emph{non-uniform} complexity, where instead of considering a Turing machine solving a problem on all input sizes, one considers a model such as boolean circuits where the computational device can change with the size of the input.  While circuits are at least as powerful as Turing machines, and can even (trivially) compute undecidable languages, their ability to solve computational problems of interest can seem closer to uniform computation.  For example, if circuits can solve $\NP$-hard problems then there are unexpected implications for uniform computation similar to $\P=\NP$ (the polynomial hierarchy collapses (\cite{KarpLipton})). As such, obtaining lower bounds for boolean circuits was seen as a viable method to indirectly tackle Turing machine lower bounds, with the benefit of being able to appeal to more combinatorial methods and thus bypassing the relativization barrier of Baker, Gill and Solovay~\cite{BGS75} which seems to obstruct most methods that can exploit uniformity.

There have been many important lower bounds obtained for restricted classes of circuits: constant-depth circuits (\cite{Ajtai83,FurstSS84,Yao85,Hastad89}), constant-depth circuits with prime modular gates (\cite{Razborov87,Smolensky87}), as well as lower bounds for monotone circuits (\cite{Razborov85, AB87,Tardos88}). Razborov and Rudich~\cite{RR97} observed that many of these lower bounds prove \emph{more} than just a lower bound for a single explicit function.  Indeed, they observed that such lower bounds often distinguish functions computable by small circuits from \emph{random} functions, and in fact they do so \emph{efficiently}.  Specifically, a \emph{natural property} $P$ is a subset of boolean functions $P\subseteq\cup_{n\ge 1}\{f:\bits^n\to\bits\}$ with the following properties, where we denote $N\eqdef 2^n$ to be the input size to the property.\footnote{The Razborov and Rudich~\cite{RR97} definition of a natural property actually applies to the complement of the property $P$ we use here. This is a trivial difference for boolean complexity, but is important for algebraic complexity as there natural properties are one-sided, see \autoref{sec:alg-nat-pf}.}
\begin{enumerate}
	\item \emph{Usefulness:} If $f:\bits^n\to\bits$ is computable by $\poly(n)$-size circuits then $f$ has property $P$.

	\item \emph{Largeness:} Random functions $f:\bits^n\to\bits$ do not have the property $P$ with noticeable probability, that is, with probability at least $1/\poly(N)=2^{-O(n)}$.

	\item \emph{Constructivity:} Given a truth-table of a function $f:\bits^n\to\bits$, of size $N=2^n$, deciding whether $f$ has the property $P$ can be checked in $\poly(N)=2^{O(n)}$ time.
\end{enumerate}
To obtain a circuit lower bound, a priori one only needs to obtain a (non-trivial) property $P$ that is useful in the above sense.  However, Razborov and Rudich~\cite{RR97} showed that (possibly after a small modification) most circuit lower bounds (such as those for constant-depth circuits (\cite{Ajtai83,FurstSS84,Yao85,Hastad89,Razborov87,Smolensky87})) yield large and constructive properties, and called such lower bounds \emph{natural proofs}. 

Further, Razborov and Rudich~\cite{RR97} argued that standard cryptographic assumptions imply that natural proofs \emph{cannot} yield super-polynomial lower bounds against any restricted class of circuits that is sufficiently rich to implement cryptography. That is, a \emph{pseudorandom function} is an efficiently computable function $f:\bits^n\times\bits^\lambda\to\bits$ such that when sampling the key $\rk\in\bits^\lambda$ at random the resulting distribution of functions $f(\cdot,\rk)$ is computationally indistinguishable from a truly random function $\rf:\bits^n\to\bits$. The existence of pseudorandom functions follows from the existence of \emph{one-way functions} (\cite{HILL,GGM}) which is essentially the weakest interesting cryptographic assumption. There are even candidate constructions of pseudorandom functions computable by polynomial-size constant-depth threshold circuits ($\TC^0$) as given by Naor and Reingold~\cite{NaorReingold97}, whose security rests on the intractability of discrete-log and factoring-type assumptions (see also Krause and Lucks~\cite{KL01}).  As such, it is widely-believed that there are pseudorandom functions, even ones computationally indistinguishable from random except to adversaries running in $\exp(\lambda^{\Omega(1)})$-time.

In contrast, Razborov and Rudich~\cite{RR97} showed that a natural proof useful against $\poly(n)$-size circuits can distinguish a pseudorandom function from a truly random function in $\poly(2^n)$-time, which would contradict the believed $\exp(\lambda^{\Omega(1)})$-indistinguishability when taking $\lambda$ to be a large enough polynomial in $n$.  That is, suppose $P$ is a natural property.  Then for a pseudorandom function $f(\cdot,\cdot)$ and each value $k\in\bits^\lambda$ of the key, the resulting function $f(\cdot,k):\bits^n\to\bits$ has a $\poly(n)$-size circuit, and has property $P$ (by usefulness).  In contrast, random functions will not have property $P$ with noticeable probability (by largeness).  As the property is constructive, this gives a $\poly(2^n)$-time algorithm distinguishing $f(\cdot,\rk)$ from a random function, as desired.

While the natural proofs barrier has proved difficult to overcome, there are results that seem to circumvent it.  For example, the barrier does not seem to apply to the lower bounds obtained for monotone circuits (\cite{Razborov85}), as there the notion of a ``random monotone function'' is not well-defined.  Further, there are results (such as Williams'~\cite{Williams14} result of $\ACC^0\ne\NEXP$) that circumvent the natural proofs barrier by incorporating techniques from uniform complexity.  Other work has demonstrated that relaxing the notion of natural proof can avoid the implications to breaking cryptography.  Chow~\cite{Chow11} has shown that \emph{almost} natural proofs (which relax largeness slightly) \emph{can} prove super-polynomial circuit lower bounds (under plausible cryptographic or complexity-theoretic assumptions).  Williams~\cite{Williams16} has shown, among other results, that some circuit lower bounds (such as for \EXP\ or \NEXP) are \emph{equivalent} to constructive (non-trivial) properties useful against small circuits, which yet have no need for any sort of largeness. Chapman and Williams~\cite{CW15} have shown that obtaining circuit lower bounds for a self-checkable problem (such as \SAT) is \emph{essentially equivalent} to obtaining a natural property against circuits that ``check their work''.  These works suggest that the exact implications of the natural proofs barrier remains not fully understood.

\paragraph{Algebraic Natural Proofs:} Algebraic circuits are one of the most natural models for computing polynomials by using addition and multiplication.  While more restricted than general (boolean) computation, proving lower bounds for algebraic circuits has proved challenging. Yet, we do not have formal barrier results for understanding the difficulty of such lower bounds. While such lower bounds are not a priori subject to the natural proofs barrier due to the formal differences in the computational model, the relevance of the ideas of natural proofs to algebraic circuits has been repeatedly asked. Aaronson-Drucker~\cite{aaronson-blog-post} as well as Grochow~\cite{Grochow15} noticed that many of the prominent algebraic circuit lower bounds (such as \cite{nis91,nw1997,Raz06,raz-yehudayoff}) are \emph{algebraically natural}, in that they obey an algebraic form of usefulness, largeness, and constructivity.

While this would seemingly then imply a Razborov and Rudich~\cite{RR97}-type barrier for existing techniques, there is a key piece missing: we have very little evidence for the existence of \emph{algebraic} pseudorandom functions.  That is, the pseudorandom functions used by Razborov and Rudich~\cite{RR97} are \emph{boolean} functions, and naive attempts to algebrize them seemingly do not yield pseudorandom polynomials. Indeed, as algebraic circuits are a computational model weaker than general computation, it is conceivable that they are too weak to implement cryptography, so that natural proofs barrier would \emph{not} apply.  In contrast, it is also conceivable that algebraic circuits are sufficiently strong so that they can compute ``enough'' cryptography to be secure against algebraic circuits, so that a natural proofs barrier \emph{would} apply.

\paragraph{Our Work:} In this work we formalize the study of pseudorandom polynomials by exhibiting the \emph{first} constructions provably secure against restricted classes of algebraic circuits. Our notion of pseudorandomness is related to the \emph{polynomial identity testing} problem, the derandomization of which is one of the main open problems in algebraic complexity theory (see \autoref{sec:alg-PRF} for more details). In particular, we follow Williams~\cite{Williams16} in treating the existence of a natural proofs barrier as the problem of \emph{succinct} derandomization: replacing randomness with pseudorandomness that further has a \emph{succinct} description.  We revisit existing derandomization of restricted classes of algebraic circuits and show (via non-trivial modification) that they can be made succinct in many cases.
 
A more formal statement of the results appears in \autoref{sec:results}. In order to present them, however, we require some technical background and definitions, which will be presented in the forthcoming sections.

Recently, and independently of our work, Grochow, Kumar, Saks, and Saraf~\cite{GKSS17} observed a similar connection between a natural proofs barrier for algebraic circuits and succinct derandomization. Their work also presents connections with Geometric Complexity Theory (which we discuss below in \autoref{sec:GCT}) and algebraic proof complexity. However, unlike our work they do not present any constructions of succinct derandomization.

\subsection{Algebraic Complexity}\label{subsec:alg-complex}

We now discuss the algebraic setting for which we wish to present the natural proofs barrier. Algebraic complexity theory studies the complexity of syntactic computation of polynomials using algebraic operations. The most natural model of computation is that of an \emph{algebraic circuit}, which is a directed acyclic graph whose leaves are labeled by either variables $x_1, \ldots, x_n$ or elements from the field $\F$, and whose internal nodes are labeled by the algebraic operations of addition ($+$) or multiplication ($\times$). Each node in the circuit computes a polynomial in the natural way, and the circuit has one or more \emph{output nodes}, which are nodes of out-degree zero. The \emph{size} of the circuit is defined to be  the number of wires, and the \emph{depth} is defined to be the length of a longest path from an input node to an output node. As usual, a circuit whose underlying graph is a tree is called a \emph{formula}. One can associate various complexity classes with algebraic circuits, and the most important one for us is $\VP$, which the classes of $n$-variate polynomials with $\poly(n)$-degree computable by $\poly(n)$-size algebraic circuits. There is also $\VNP$, which we will informally define as the class of ``explicit'' polynomials.

A central open problem in algebraic complexity theory is to prove a super-polynomial lower bound for the algebraic circuit size of any explicit polynomial, that is, proving $\VP\ne\VNP$.  Substantial attention has been given to this problem, using various techniques that leverage non-trivial algebraic tools to study the syntactic nature of these circuits.  Indeed, our knowledge of algebraic lower bounds seem to surpass that of boolean circuits, as we have super-linear lower bounds for general circuits (\cite{Str73, BS83}) --- a goal as yet unachieved in the boolean setting.  Similarly, there are a wide array of super-polynomial or even exponential lower bounds known for various weaker models of computation such as non-commutative formulas (\cite{nis91}), multilinear formulas (\cite{raz2004, ry08}), and homogeneous depth-3 and depth-4 circuits (\cite{nw1997, GKKS16, KSS13, FLMS15, KLSS, KS14}). We refer the reader to Saptharishi~\cite{github} for a continuously-updating comprehensive compendium of these lower bounds.

However, this landscape might still feel reminiscent of the boolean setting, in that there are various restricted models where lower bounds techniques are known, and yet lower bounds for general circuits or formulas remain relatively poorly understood.  Yet, there has been some significant recent cause for optimism for obtaining \emph{general} circuit lower bounds, as various depth-reduction results (\cite{vsbr83,ajmv98,av08,Koiran,Tav15,GKKS16,CKSV16}) have shown that $n$-variable degree-$d$ polynomials computable by size-$s$ algebraic circuits have $s^{O(\sqrt{d})}$-size depth-3 or homogeneous depth-4 formulas. Further, recent methods (\cite{kayal12,GKKS14, KSS13, FLMS15, KLSS, KS14}) have proven $(nd)^{\Omega(\sqrt{d})}$ lower bounds computing explicit polynomials by homogeneous depth-4 formulas. If one could simply push these methods to obtain an $(nd)^{\omega(\sqrt{d})}$ lower bound then this would obtain super-polynomial lower bounds for general circuits!  Unfortunately, all of the lower bounds methods known seem to apply not just to candidate hard polynomials, but also to certain \emph{easy} polynomials, demonstrating that these techniques cannot yield a $(nd)^{\omega(\sqrt{d})}$ lower bound as this would contradict the depth-reduction theorems.

Given this state of affairs, it is unclear whether to be optimistic or pessimistic regarding future prospects for obtaining superpolynomial lower bounds for general algebraic circuits.  To resolve this uncertainty it is clearly important to formalize the barriers constraining our lower bound techniques.  Indeed, as mentioned above all known lower-bound methods apply not just to hard polynomials but also to easy polynomials --- is this intrinsic to current methods?  This is essentially the question of whether there is an algebraic natural proofs barrier, as we now describe.

\subsection{Algebraic Natural Proofs}\label{sec:alg-nat-pf}

We now define the notion of an \emph{algebraically natural proof} used in this paper.  Intuitively, we want to know whether lower bounds methods can distinguish between low-complexity and high-complexity polynomials, so that they are \emph{useful} in the sense of \RR. In particular, we want to know if such distinguishers\footnote{Grochow~\cite{Grochow15} referred to distinguishers as \emph{test polynomials}, as they test whether an input polynomial is of low- or high-complexity.} can be \emph{efficient}, so that they are also \emph{constructive}.  Several works, such as Aaronson and Drucker~\cite{aaronson-blog-post}, Grochow~\cite{Grochow15} (see also Shpilka and Yehudayoff~\cite[Section 3.9]{sy}, and Aaronson~\cite[Section 6.5.3]{Aar16}) have noticed that almost all of the lower bounds methods in algebraic complexity theory are themselves algebraic in a certain sense which we now describe.

The simplest example is to consider matrix rank, where the complexity of an $n\times n$ matrix $M$ is exactly captured by its determinant, which is a polynomial. That is, if $M$ is of rank $<n$ then $\det M=0$, and if rank $=n$ then $\det M\ne 0$. The key feature here is that $\det M$ is a polynomial in the \emph{coefficients of the underlying algebraic object}, which in this case is the matrix $M$. Most of the central lower bounds techniques, such as partial derivatives (\cite{nw1997}), evaluation/coefficient dimension (\cite{nis91,Raz06,raz-yehudayoff,FS13}), or shifted partial derivatives (\cite{kayal12,GKKS14}) are generalizations of this idea, specifically leveraging notions of linear algebra and rank. Abstractly, these methods take an $n$-variate polynomial $f$, inspect its coefficients, and then form an exponentially-large (in $n$) matrix $M_f$ whose entries are polynomials in the coefficients of $f$. One then shows that if $f$ is simple then $\rank M_f < r$, while for an explicit polynomial $h$ one can show that $\rank M_{h}\ge r$.  In particular, by basic linear algebra this shows that there is some $r\times r$ submatrix $M'_{h}$ of $M_{h}$ such that $\det M'_{h}\ne 0$, and yet $\det M'_f=0$ for simple $f$, where $M'_f$ denotes the restriction of $M_f$ to the same set of rows and columns. This proves that $h$ is a hard polynomial.

We now observe that the above outline gives a natural property $P\eqdef\{f:\det M_f'=0\}$ in the sense of \RR.

\begin{enumerate}
	\item \emph{Usefulness:} For low-complexity $f$ we have that $f\in P$ as argued above.  Further, $P$ is a non-trivial property as $h\notin P$.
		
	\item \emph{Constructivity:} For a given $f$, deciding whether ``$f\in P$?'' is tantamount to computing $\det M'_f$.  Even though $M'_f$ might be exponentially-large, it is often polynomially-large in the \emph{size of $f$} (which is exponential in the number $n$ of variables in $f$).  As typically $M'_f$ is a simple matrix in terms of $f$, computing $\det M'_f$ is essentially the complexity of computing the determinant, which is computable by small algebraic circuits (\cite{Berk84,mv97}).  Thus, the property $P$ is efficiently decidable in the size of its input.
		
	\item \emph{Largeness:} The largeness condition is \emph{intrinsic} here, as the property is governed by the vanishing of a non-zero polynomial; $\det M'_f$ is non-zero as a polynomial as in particular $\det M'_{h}\ne 0$. As non-zero polynomials evaluate to non-zero at random points with high probability (\cite{S80, Z79, DL78}), this means that such distinguishers certify that random polynomials are of high-complexity.
\end{enumerate}

Thus, we see that the above meta-method forms a very natural instance of a natural property.  As such, one might expect the \RR barrier to then rule out such properties, however their barrier result only holds when the underlying circuit class can compute pseudorandom functions. While it is widely believed that simple boolean circuit classes can compute pseudorandom functions (as discussed above), the ability of algebraic circuits to compute pseudorandom functions is significantly less understood. As such, the \RR barrier's applicability to the algebraic setting is not immediate.  However, as the above meta-method obeys algebraic restrictions on the natural properties being considered, this suggests that barrier could follow from a weaker assumption than that of algebraic circuits computing pseudorandom functions. 

We now give a formalization of the above meta-method for algebraic circuit lower bounds, which is implicit in prior work and known to experts. To begin, we must first note that in comparing low-complexity to high-complexity polynomials, we must detail the space in which the polynomials reside.  There are three spaces of primary interest.
\begin{enumerate}
	\item $\totdeg$: The space of $n$-variate polynomials of total degree at most $d$. There are $\totdegN\eqdef\binom{n+d}{d}$ many monomials $\vx^\va\eqdef x_1^{a_1}\cdots x_n^{a_n}$ in this space.

	\item $\totdeghom$: The space of homogeneous $n$-variate polynomials of total degree exactly $d$. There are $\totdeghomN\eqdef\binom{n+d-1}{d}$ many monomials $\vx^\va$ in this space.

	\item $\inddeg$: The space of $n$-variate polynomials of individual degree at most $d$. There are $\inddegN\eqdef(d+1)^n$ many monomials $\vx^\va$ in this space.
\end{enumerate}
While this may seem pedantic, it is important to distinguish these spaces.  That is, while homogeneous degree-$d$ polynomials capture nearly all of the interesting complexity of polynomials of degree \emph{at most} $d$, it is trivial to distinguish the two. That is, consider the distinguisher polynomial $c_{\vec{0}}$ that simply returns the constant coefficient (the coefficient of $1$) of a polynomial $f=\sum_{\va}c_\va\vx^\va$.  This polynomial vanishes on \totdeghom for $d>0$, but does not vanish on the constant polynomial $1\in\totdeg$.  However, it would be absurd to say that ``$1$ is a hard polynomial for \totdeghom''.  Thus, in discussing how properties can distinguish polynomials we must specify the domain of interest. Indeed, to discuss lower bounds for homogeneous computation one must restrict attention to the space \totdeghoms, and likewise to discuss lower bounds for multilinear computation one must restrict attention to the space \inddegsm.

We now present our definition, with enough generality to handle the above spaces of polynomials simultaneously.  That is, for a fixed set of monomials $\cM$ (such as all monomials of degree at most $d$) we consider the space $\Span(\cM)$, which is defined as all linear combinations over monomials in $\cM$.  We then identify a polynomial $f\in\Span(\cM)$ defined by $f=\sum_{\vx^\va\in \cM} c_\va\vx^\va$ with its list of such coefficients, which is a vector $\coeff_{\cM}(f)\in\F^{\cM}$ defined $\coeff_{\cM}(f)\eqdef (c_\va)_{\vx^\va\in \cM}$.  We then ask for distinguisher $D$ which take as input these $|\cM|$ many coefficients, which can separate low-complexity polynomials from high-complexity polynomials.

\begin{definition}[Algebraically Natural Proof]\label{defn:alg-nat-pf}
	Let $\cM\subseteq\F[x_1,\ldots,x_n]$ be a set of monomials $\cM=\{\vx^\va\}_\va$, and let the set $\Span(\cM)\eqdef\{\sum_{\vx^\va\in \cM} c_\va\vx^\va: c_\va\in\F\}$ be all linear combinations of these monomials.  Let $\cC\subseteq\Span(\cM)$ and $\cD\subseteq \F[\{c_\va\}_{\vx^\va\in \cM}]$ be classes of polynomials, where the latter is in $|\cM|$ many variables. 
	
	A polynomial $D\in\cD$ is an \demph{algebraic $\cD$-natural proof against $\cC$}, also called a \demph{distinguisher}, if
	\begin{enumerate}
		\item $D$ is a non-zero polynomial.\label{defn:alg-nat-pf:nz}
		\item For all $f\in\cC$, $D$ vanishes on the coefficient vector of $f$, that is, $D(\coeff_{\cM}(f))=0$.\label{defn:alg-nat-pf:vanish}
			\qedhere
	\end{enumerate}
\end{definition}

We will be primarily interested in taking the set of monomials $\cM$ to correspond to one of the above three sets of polynomials, \totdegs, \totdeghoms and \inddegs, to which we define the relevant coefficient vectors as \coefftot, \coefftothom and \coeffind. We will use ``$\coeff$'' if the space of polynomials is clear from the context.

Thus, to revisit the comparison with \RR, condition \eqref{defn:alg-nat-pf:vanish} says that the distinguisher $D$ is \emph{useful} against the class $\cC$.  Condition \eqref{defn:alg-nat-pf:nz} indicates that the property is non-trivial, and in particular is \emph{large}, as a non-zero polynomial will evaluate to non-zero at a random point with high probability (\cite{S80, Z79, DL78}).  Finally, the fact that distinguisher $D$ comes from the restricted class $\cD$ is the \emph{constructivity} requirement, and the main question is how simple the distinguisher $D$ can be.

Further, note how the above distinguishers naturally have a \emph{one-sided} nature to them as in algebraic complexity one typically seeks lower bounds against computations using \emph{any} field extension of the base field of coefficients.  In using the above to define the \RR style property $P\eqdef\{f:D(\coeff_{\cM}(f))=0\}$, we note that the complement property $\neg P=\Span(\cM)\setminus P=\{f:D(\coeff_{\cM}(f))\ne 0\}$ \emph{cannot} be expressed in the above framework.  That is, for non-zero polynomials $p$ and $q$, it cannot be that the product $pq$ vanishes everywhere (over large enough fields), so that in particular it cannot be that $p(\vaa)=0$ iff $q(\vaa)\ne 0$.

We argued above that most of the main lower bound techniques fall into the above algebraic natural proof paradigm where the distinguisher has polynomial-size algebraic circuits, so that the proof is $\VP$-natural. This motivates the following question about algebraic $\VP$-natural proofs against $\VP$.

\begin{question}
	For the space of total degree $d$ polynomials $\totdeg$, is there an algebraic $\poly(\totdegN)$-size natural proof for lower bounds against $\poly(n,d)$-size circuits?
\end{question}

While one could make a detailed study of existing lower bounds to prove the intuitive fact that $\VP$-natural properties suffice for them, our attention will be to studying the \emph{limits} of this framework.  That said, it is worth mentioning that there are known techniques for algebraic circuit lower bounds that fall outside this framework.  

First, the shifted partial derivative technique of Gupta, Kamath, Kayal and Saptharishi~\cite{kayal12,GKKS14} is not currently known to be $\VP$-natural.  That is, while it does fall into the above rank-based meta-method (and thus the algebraic natural proof paradigm), the matrices involved are actually \emph{quasi}-polynomially large in their input, so the method is only $\quasi\VP$-natural.  However, as the shifted partial technique proves exponential lower bounds the required $\quasi\VP$-naturalness still seems rather modest.  

In contrast, there are actually methods which \emph{completely} avoid the algebraic framework (constructive or not).  That is, as discussed below in \autoref{sec:GCT}, this algebraic distinguisher framework is limited to proving \emph{border} complexity lower bounds, where border complexity is always upper bounded by usual complexity notions.  For the \emph{tensor rank} model, distinguishers actually prove border rank lower bounds. In contrast, the substitution method (\cite[Chapter 6]{bcs97}, \cite{Blaser14}) can prove tensor rank lower bounds which are \emph{higher} than known border rank upper bounds (for explicit tensors), giving a separation between these two complexities and thus showing the substitution method is not captured by the algebraic natural proof framework.  However, all such known separations are by at most a multiplicative constant factor, so the inability of the substitution method to be algebraically natural does not currently seem to be a serious deficiency in the framework developed here.

\subsection{Pseudorandom Polynomials}\label{sec:alg-PRF}

Having given our formal definition of algebraic natural proofs, we now explain our notion of the algebraic natural proof \emph{barrier}. In particular, as algebraically natural proofs concern the zeros of (non-zero) polynomials computable by small circuits, this naturally leads us to the \emph{polynomial identity testing (PIT)} problem.

\paragraph{Polynomial Identity Testing:}

Polynomial identity testing is the following algorithmic problem: given an algebraic circuit $D$ computing an $N$-variate polynomial, decide whether $D$ computes the identically zero polynomial. The problem admits a simple efficient randomized algorithm by the Schwartz-Zippel-DeMillo-Lipton Lemma~\cite{S80, Z79, DL78}. That is, evaluations of a low-degree non-zero polynomial at random points taken from a large enough field will be non-zero with high probability. Thus, to check non-zeroness it is enough to evaluate $D$ on a random input $\vecalpha$ and observe whether $D(\vecalpha) = 0$, which is clearly efficient. However, the best known deterministic algorithms run in exponential time. Designing an efficient deterministic algorithm for PIT is another major open problem in algebraic complexity, with intricate and bidirectional connections to proving algebraic and boolean circuit lower bounds \cite{HS80, a05, ki03, dsy09}.

The two flavors in which the problem appears are the \emph{white-box} model, in which the algorithm is allowed to inspect the structure of the circuit, and the \emph{black-box} model, in which the algorithm is only allowed to access evaluations of the circuit on inputs of its choice, such as the randomized algorithm described above. It can be easily seen that efficient deterministic black-box algorithms are equivalent to constructing small \emph{hitting sets}: a hitting set for a class $\cD\subseteq\F[c_1,\ldots,c_N]$ of circuits is a set $\cH \subseteq \F^N$ such that for any non-zero circuit $D \in \cD$, there exists $\vecalpha \in\cH$ such that $D(\vecalpha) \neq 0$.  While small hitting sets \emph{exist} for $\VP$, little progress has been made for explicitly constructing any non-trivial hitting sets for general algebraic circuits (or even solving PIT in the white-box model).  In contrast, there has been substantial work developing efficient deterministic white- and black-box PIT algorithms for non-trivial restricted classes of algebraic computation. For more, see the surveys of Saxena~\cite{Saxena09,Saxena14} and Shpilka-Yehudayoff~\cite{sy}.

\paragraph{Succinct Derandomization:} 

We now define our notion of pseudorandom polynomials by connecting the algebraic natural proof framework with hitting sets.  Consider a class $\cC$ of polynomials, say within the space of polynomials of bounded total degree $\totdeg$.  If $D$ is an algebraic natural proof against $\cC$ then we have:
\begin{enumerate}
	\item $D$ is a non-zero polynomial.

	\item $D$ vanishes on the set $\cH\eqdef\{\coefftot(f):f\in\cC\}$ of coefficient vectors of polynomials in $\cC$.
\end{enumerate}
Put together, these conditions are equivalent to saying that that $\cH$ is \emph{not} a hitting set for $D$.  Thus, we see that there are algebraically natural proofs \emph{if and only if} coefficient-vectors of simple polynomials are \emph{not} hitting sets.  In other words, the existence of an algebraic natural proofs barrier can be rephrased as whether PIT can be derandomized using \emph{succinct} pseudorandomness.  A completely analogous statement was proven by Williams~\cite{Williams16} in the boolean setting, where the existence of the \RR natural proofs barrier was shown equivalent to succinct derandomization of $\ZPE$, those problems solvable in zero-error $2^{O(n)}$-time.  However, that equivalence there is slightly more involved, while it is immediate here.

We now give the formal definition mirroring the above discussion, in the same generality of \autoref{defn:alg-nat-pf}.

\begin{definition}[Succinct Hitting Set]\label{defn:alg-PRF}
	Let $\cM\subseteq\F[x_1,\ldots,x_n]$ be a set of monomials $\cM=\{\vx^\va\}_\va$, and let the set $\Span(\cM)\eqdef\{\sum_{\vx^\va\in \cM} c_\va\vx^\va: c_\va\in\F\}$ be all linear combinations of these monomials.  Let $\cC\subseteq\Span(\cM)$ and $\cD\subseteq \F[\{c_\va\}_{\vx^\va\in \cM}]$ be classes of polynomials, where the latter is in $|\cM|$ many variables. 
	
	$\cC$ is a \demph{$\cC$-succinct hitting set for $\cD$} if $\cH\eqdef\{\coeff_\cM(f):f\in\cC\}$ is a hitting set for $\cD$.  That is, $D\in\cD$ is non-zero iff $D|_\cH$ is non-zero, that is, there is some $f\in\cC$ such that $D(\coeff_\cM(f))\ne 0$.
\end{definition}

To make our statements more concise, we often abbreviate the name of the class $\cC$ in a way which is understood from the context. For example, the modifier ``$s$-succinct'', with $s$ being an integer, will refer to a $\cC$-hitting set with $\cC$ being the class of circuits of size at most $s$. Similarly, $s$-$\SPS$-succinct will refer to $\cC$ being the class of depth-3 circuits of size at most $s$, and so on.

The above argument showing the tension between algebraic natural proofs and pseudorandom polynomials can be summarized in the following theorem, which follows immediately from the definitions.

\begin{theoremwp}\label{thm:natural-proofs-implies-no-succinct-hitting-set}
	Let $\cM\subseteq\F[x_1,\ldots,x_n]$ be a set of monomials $\cM=\{\vx^\va\}_\va$, and let the set $\Span(\cM)\eqdef\{\sum_{\vx^\va\in \cM} c_\va\vx^\va: c_\va\in\F\}$ be all linear combinations of these monomials.  Let $\cC\subseteq\Span(\cM)$ and $\cD\subseteq \F[\{c_\va\}_{\vx^\va\in \cM}]$ be classes of polynomials, where the latter is in $|\cM|$ many variables.
	
	Then there is an algebraic $\cD$-natural proof against $\cC$ iff $\cC$ is not a $\cC$-succinct hitting set for $\cD$.
\end{theoremwp}

Instantiating this claim with $\cM$ being the space of degree-$d$ monomials, we get the following quantitative version of the above.

\begin{corollarywp}
	Let $\cC\subseteq\totdeg$ be the class of $\poly(n,d)$-size circuits of total degree at most $d$.  Then there is an algebraic $\poly(\totdegN)$-natural proof against $\cC$ iff $\cC$ is not a $\poly(n,d)$-succinct hitting set for $\poly(\totdegN)$-size circuits in $\totdegN$ variables.
\end{corollarywp}

In the common regime when $d=\poly(n)$, we have that $\poly(n)=\polylog(\totdegN)$.  That is, this existence of an algebraic natural proofs barrier is equivalent to saying that coefficient-vectors of polylogarithmic-size circuits (in polylogarithmic many variables) form a hitting set of polynomial-size.

With this equivalence in hand, we can now phrase the question of an algebraic natural proofs barrier.

\begin{question}[Algebraic Natural Proofs Barrier]\label{q:barrier}
	Is there a $\polylog(N)$-succinct hitting set for circuits of $\poly(N)$-size?
\end{question}

Again, we note that \autoref{q:barrier} was also raised by Grochow, Kumar, Saks, and Saraf~\cite{GKSS17}, who presented a definition similar to \autoref{defn:alg-PRF} and also observed the implication in \autoref{thm:natural-proofs-implies-no-succinct-hitting-set}.

\paragraph{Succinct Generators:} 

While the above equivalence already suffices for studying the barrier, the notion of a hitting set is sometimes fragile. A more robust way to obtain hitting sets for a class $\cD\subseteq\F[c_1,\ldots,c_N]$ is to obtain a \emph{generator}, which is a polynomial map $\cvG:\F^\ell\to\F^N$ such that $D\in\cD$ is a non-zero iff $D\circ \cvG\not\equiv 0$, that is, the composition $D(\cvG(\vy))\ne 0$ is non-zero as a polynomial in $\vy$.  Here one measures the quality of the generator by asking to minimize the seed-length $\ell$.  By polynomial interpolation, it follows that constructing small hitting sets is equivalent to constructing generators with $\ell$ small, see for example Shpilka-Yehudayoff~\cite{sy}.  

However, in our setting we want \emph{succinct} generators so that the polynomial-map $\cvG$ is a coefficient vector of a polynomial $G(\vecx, \vecy)$ computable by a small algebraic circuit.  In particular, converting a succinct hitting set $\cH$ to a generator using the standard interpolation methods would give a generator which has circuit size $\poly(|\cH|)$.  However, as we are trying to hit polynomials on $N$ variables, this would yield a $\poly(N)$-size generator whereas we would want a generator of complexity $\polylog(N)$.  As such, we now define succinct generators and give a tighter relationship with succinct hitting sets.

\begin{definition}\label{defn:succ-gen}
	Let $\cM\subseteq\F[x_1,\ldots,x_n]$ be a set of monomials $\cM=\{\vx^\va\}_\va$, and let the set $\Span(\cM)\eqdef\{\sum_{\vx^\va\in \cM} c_\va\vx^\va: c_\va\in\F\}$ be all linear combinations of these monomials.  Let $\cC\subseteq\Span(\cM)$ and $\cD\subseteq \F[\{c_\va\}_{\vx^\va\in \cM}]$ be classes of polynomials, where the latter is in $|\cM|$ many variables.  Further, let $\cC'\subseteq\F[x_1,\ldots,x_n,y_1,\ldots,y_\ell]$ be another class of polynomials.

	We say that a polynomial map $\cvG:\F^\ell\to\F^\cM$ is a \demph{$\cC$-succinct generator for $\cD$ computable in $\cC'$} if
	\begin{enumerate}
		\item The polynomial $G(\vx,\vy)\eqdef \sum_{\vx^\va\in\cM} \cG_{\vx^\va}(\vy)\cdot\vx^\va$ is a polynomial in $\cC'$, where $\cG_{\vx^\va}(\vy)$ is the polynomial computed by the $\vx^\va$-coordinate of $\cvG$.\label{defn:succ-gen:indexing}
		\item For every value $\vaa\in\F^\ell$, the polynomial $G(\vx,\vaa)\in\cC$.\label{defn:succ-gen:output}
		\item $\cvG$ is a generator for $\cD$. That is, $D\in\cD$ is a non-zero polynomial in $\F[\vc]$ iff $D\circ \cvG\not\equiv 0$ in $\F[\vy]$, meaning that $D(\coeff_\cM (G(\vx,\vy)))\ne 0$ as a polynomial in $\F[\vy]$, where we think of $G(\vx, \vy)$ as a polynomial in the ring $\left( \F[\vy] \right) [\vx]$  and take these coefficients with respect to the $\vx$ variables, so that $\coeff_\cM (G(\vx,\vy))\in\F[\vy]^\cM$.\label{defn:succ-gen:gen}
			\qedhere
	\end{enumerate}
\end{definition}

Conditions \eqref{defn:succ-gen:output} and \eqref{defn:succ-gen:gen} are equivalent, over large enough fields, to the property that the output of the generator $\cvG(\vx,\F^\ell)=\{G(\vx,\vaa): \vaa\in\F^\ell\}$ is a $\cC$-succinct hitting set for $\cD$. However, the generator result is a priori stronger as it says that the hitting set can be succinctly indexed by a polynomial in $\cC'$.  

Also, note that the $\cC'$ computability of the generator implies $\cC'$-succinctness, that is, that its image $\{G(\vx,\vaa):\vaa\in\F^\ell\}$ are all circuits which are $\cC'$-circuits, at least assuming that $\cC'$ is a class of polynomials which is closed under substitution.  However, sometimes the actual succinctness $\cC$ can be more stringent than $\cC'$ for restricted classes of computation. Since the implication regarding barriers to lower bounds only concerns the class $\cC$, we often omit mentioning $\cC'$ and only talk about $\cC$-succinct generators for $\cD$.

This definition bears a slight resemblance to Mulmuley's~\cite{Mulmuley12} definition of an ``explicit variety''. A discussion about the connections between our work and Geometric Complexity Theory appears in \autoref{sec:GCT}.

We now give our first result, which uses the construction of a universal circuit to show that there is an explicit universal construction of a succinct generator, that is, this circuit is a succinct generator if there are \emph{any} succinct hitting sets.  Further, this shows that \emph{any} succinct hitting set (even infinite) implies a quasipolynomial deterministic black-box PIT algorithm.  To make this theorem clear, let $\VP_m$ denote the class of small low-degree circuits in $m$ variables.

\begin{theorem*}[Informal summary of \autoref{sec:universal}]
	There is an explicit $\polylog(N)$-size circuit which is a $\VP_{\polylog(N)}$-succinct generator for $\VP_N$ iff there is a $\VP_{\polylog(N)}$-succinct hitting set for $\VP_N$. Further, the existence of any $\VP_{\polylog(N)}$-succinct hitting set for $\VP_N$ implies an explicit $\poly(N)^{\polylog(N)}$-size hitting set for $\VP_N$.
\end{theorem*}

Note that Aaronson and Drucker~\cite{aaronson-blog-post} proposed a candidate algebraic pseudorandom function based on generic projections of determinants. Their construction does not seem sufficient for the above result, as discussed in \autoref{sec:universal}.

\subsection{Evidence for Pseudorandom Polynomials and Our Results}

Having now given our formalization of algebraic natural proofs and the corresponding barrier, we now investigate evidence for such barriers.  To understand these barriers, it is helpful to remind ourselves of the evidence in the boolean setting.

\paragraph{Boolean Complexity:} When speaking of a natural proofs barrier, it is helpful to remember that such barriers are inherently \emph{conditional} (as opposed to relativization (\cite{BGS75}) and algebraization (\cite{AW09}), which are unconditional).  As such, our belief in such barriers rests on the plausibility of these conditional assumptions. We now review two sources of evidence, cryptographic and complexity-theoretic.

The \RR paper showed that there is a natural proofs barrier under the assumption of the existence of pseudorandom functions with exponential security.  As discussed in the introduction, there are two good reasons to believe the plausibility of this assumption.  First, is that there are many well-studied candidate constructions which are believed to have this security.  Second, is that there is a web of security-preserving reductions between cryptographic notions, in particular showing that such pseudorandom functions follow from pseudorandom generators with exponential security (\cite{GGM}) or even one-way functions with exponential security (\cite{HILL}).  One-way functions are the most basic cryptographic object, so that essentially the natural proofs barrier holds unless cryptography fails.\footnote{Furthermore, there exist problems, such as the discrete logarithm problem, for which the natural proof barrier holds unconditionally.}

The above cryptographic evidence already seems strong enough, but it is worth mentioning another evidence based on complexity-theoretic derandomization.  That is, for many classes of restricted computation there have been pseudorandom generators $\cvG:\bits^\ell\to\bits^n$ that fool these restricted classes even when $\ell=\polylog(n)$. For example, $\AC^0$ is fooled by the Nisan-Wigderson~\cite{nw94} generator instantiated with parity (\cite{Nisan91b}) as well as by $\polylog(n)$-wise independence (\cite{Braverman10}), $\RL$ is fooled by Nisan's~\cite{Nisan92} generator, and $\eps$-bias spaces fool linear polynomials over $\F_2$ (\cite{NaorNaor93,AGHP92}).  In each of these cases it turns out that the generators are in fact pseudorandom \emph{functions} that fool these restricted classes, in that for every seed $\ell$, $\cvG(\ell)$ can be thought of as a truth table of a function $\cG_\ell:\bits^{\log n}\to\bits$, such that $\cG_\ell$ can actually be computed in $\poly(\ell,\log n)=\polylog(n)$-time (as $\ell=\polylog(n)$).  That is, these derandomization results actually provide succinct derandomization in the sense of Williams~\cite{Williams16}. In fact, \RR explicitly noted how Nisan's~\cite{Nisan91b} pseudorandom generator for $\AC^0$ is a pseudorandom function (with an application to how the lower bounds for $\AC^0[2]$ are thus provably more complicated than those for $\AC^0$). It can be seen that fooling a restricted class of computation $\cC$ using the Nisan-Wigderson~\cite{nw94} generator, when the hard function $f$ against $\cC$ is actually efficiently computable in $\mathsf{P}$, gives rise to a pseudorandom function \emph{unconditionally} secure against $\cC$. In this case, each output of the Nisan-Wigderson generator can be thought of as a truth table of a simple function; this follows from the assumption on $f$ and the fact that designs are efficiently computable, and see \cite{CIKK16} for further discussion.

\paragraph{Algebraic Complexity:} Having reviewed the evidence for a natural proofs barrier in the boolean setting, we can then ask: what evidence is there for an algebraic natural proofs barrier?  Unfortunately, such evidence has been much more difficult to obtain.

Indeed, the cryptographic evidence in the boolean setting seems less relevant to the algebraic world.  Direct attempts to algebrize the underlying cryptographic objects will only yield \emph{functions} that seem pseudorandom, where as we need \emph{polynomials}. While our universal construction (\autoref{sec:universal}) gives a universal candidate pseudorandom polynomial, we lack the corresponding web of reductions that reduces the analysis of such candidates to more traditional and well-studied conjectures.  In particular, the construction of Goldreich, Goldwasser and Micali~\cite{GGM} that converts a pseudorandom generator to a pseudorandom function seems to have no algebraic analogue (\cite{aaronson-blog-post}) as this construction applied to polynomials produces polynomials of exponential degree and thus do not live in the desired space of low-degree polynomials $\totdeg$. 

Given the complete lack of algebraic-cryptographic evidence for an algebraic natural proofs barrier, it is then natural to turn to \emph{complexity-theoretic} evidence in the form of succinct derandomization, which constitutes our results.

\subsection{Our Results}\label{sec:results}

In this work we present the first unconditional succinct derandomization of various restricted classes of algebraic computation, giving the first evidence at all for an algebraic natural proofs barrier.  It is worth noting that in the boolean setting, as discussed above, many derandomization results are \emph{already} succinct.  It turns out that, to the best of our knowledge, all existing derandomization for restricted algebraic complexity classes are \emph{not} succinct. 

A primary reason for this is that to obtain the best derandomization for polynomials, one typically wants to use univariate generators as this produces more randomness-efficient results (much in the same way that univariate Reed-Solomon codes have better distance than multi-variate Reed-Muller codes). However, univariate polynomials are not $\VP$-succinct essentially by definition as $\VP$ looks for multivariate polynomials where the degree is commensurate with the number of variables. Another reason is that while hardness-vs.-randomness can produce succinct derandomization in the boolean setting as mentioned above, the known algebraic hardness-vs.-randomness paradigm (\cite{ki03}) is much harder to instantiate for restricted classes of algebraic computation.

However, it seems highly plausible that by redoing existing constructions one can obtain succinct derandomization, and as such we posit the following meta-conjecture. 

\begin{metaconjecture}\label{conj:complexity-barrier}
	For any restricted class $\cD\subseteq\F[c_1,\ldots,c_N]$ for which explicit constructions of subexponential-size hitting sets are currently known, there are subexponential-size hitting-sets which are $\polylog(N)$-succinct, where succinctness is measured with respect to one of the spaces of polynomials $\totdeg$, $\totdeghom$, or $\inddeg$.
\end{metaconjecture}

In this work we establish this meta-conjecture for many, but not all, known derandomization results for restricted classes of algebraic circuits.  We obtain succinctness with respect to computations in the space of multilinear polynomials \inddegm. In some cases similar results could be obtained with respect to the space of total degree \totdeg, but we omit discussion of these techniques as the \inddegm results are cleanest. All of our succinct derandomization results will be via succinct generators, but as the hitting sets have succinctness even beyond the succinctness of the generator we will focus on presenting the succinctness of the hitting sets instead.

We now list our results, but defer the exact definitions of these models to the relevant sections. We begin with succinct derandomization covering many of the hitting-set constructions for constant-depth circuits with various restrictions. These formulas will be fooled by hitting sets which are themselves depth-3 formulas, but of polylogarithmic complexity.

\begin{theorem}\label{thm:succinct-results}
	In the space of multilinear polynomials \inddegm, the set of $\poly(\log s,n)$-size multilinear $\SPS$ formulas is a succinct hitting set for $N=2^n$-variate size-$s$ computations of the form
	\begin{itemize}
		\item $\Sigma^{O(1)}\Pi\Sigma$ formulas (\autoref{subsec:spsk})
		\item $\SPS$ formulas of transcendence degree $\le O(1)$ (\autoref{subsec:trdeg})
		\item Sparse polynomials (\autoref{subsec:sparse})
		\item $\SMESP{O(1)}$-formulas (\autoref{subsec:smesp})
		\item Commutative roABPs (\autoref{subsec:comm-roabp})
		\item Depth-$O(1)$ Occur-$O(1)$ formulas (\autoref{sec:occurk})
		\item Arbitrary circuits composed with sparse polynomials of transcendence degree $O(1)$ (\autoref{subsec:gen-sparse-low-trans}).
	\end{itemize}
\end{theorem}

We now conclude with a weaker result, which is not truly succinct in that the hitting set is of complexity commensurate with the class being fooled.  However, this result is for fooling classes of algebraic computation which while restricted, go beyond constant-depth formulas, and as such our result is still non-trivial. This class of computation is known as \emph{read-once oblivious algebraic branching programs (roABPs)}, which can be seen as an algebraic version of $\RL$. 

\begin{theorem}[\autoref{subsec:roabp}]
	In the space of multilinear polynomials \inddegm, the set of width-$w^2$ length-$n$ roABPs is a succinct hitting set for width-$w$ and length-$N=2^n$ roABPs with a monomial compatible ordering of the variables.
\end{theorem}

As commented above, while the \emph{length} of the roABPs whose coefficient vectors define the hitting set is merely $n=\log(N)$, the \emph{width} is as large as $w^2$, while a truly succinct hitting set would require the width to be $\polylog (w)$.

\subsection{Techniques}

We now discuss the techniques we use to obtain our succinct hitting sets.  The first technique is to carefully choose \emph{which} existing hitting sets constructions to make succinct.  In particular, one would naturally want to start with the simplest restricted classes of circuits to fool, which would be sparse polynomials.  A well-known hitting-set construction is due to Klivans and Spielman~\cite{ks01}, which is often used in hitting-set constructions for more sophisticated algebraic computation.  However, as we explain in \autoref{sec:discussion}, it actually seems difficult to obtain a succinct version of this hitting set (or variants of it).

Instead, we observe that, due to the results of \autoref{sec:universal} mentioned above, we need not focus on the \emph{size} of the hitting sets but rather only on their succinctness.  That is, to obtain succinct hitting sets for $s$-sparse polynomials we need not look at the $\poly(s)$-size hitting sets of Klivans and Spielman~\cite{ks01} but can also consider $\poly(s)^{\polylog(s)}$-size hitting sets which may be more amenable to being made succinct.  In particular, there is a generator of Shpilka and Volkovich~\cite{SV15} which can be seen as an algebraic analogue of $k$-wise independence.  It has been shown that this generator fools sparse polynomials with a hitting set of $\poly(s)^{\polylog(s)}$-size, and we show how to modify this result so the generator is also succinct.  Similarly, there is a family of hitting sets which use the \emph{rank condensers} of Gabizon and Raz~\cite{GR08} to produce a pseudorandom linear map that reduces from $n$ variables down to $r\ll n$ variables.  We also suitably modify this construction to be succinct.  Between these two core constructions, as well as their combination, we are able to make succinct much of the existing hitting set literature.

We now briefly illustrate the simplest example of how we take existing constructions and make them succinct.  Suppose one wanted to hit a non-zero linear polynomial $D(\vc)=\alpha_1c_1+\cdots+\alpha_Nc_N$.  A standard approach would be to replace $c_i\leftarrow z^i$ where $z$ is a new variable, as one now obtains a univariate polynomial $D(z)=\alpha_1z^1+\cdots+\alpha_Nz^N$ which is clearly still non-zero.  Now, however, there is simply one variable of degree $N$ so that interpolation over this variable yields a hitting set of size $N+1$, which is essentially optimal in terms of hitting set size.  To see how to make this succinct, note that the resulting vectors in the hitting set have the form $(\beta,\beta^2,\ldots,\beta^N)$ for $\beta\in\F$.  For $N=2^n$ so that we can identify $\F^N$ as the coefficient vectors of multilinear polynomials $\F[x_0,\ldots,x_{n-1}]_{\ideg}^1$, we can see that such vectors can be succinctly represented as the coefficients of $\beta(1+x_0\beta^{2^0})(1+x_1\beta^{2^1})\cdot(1+x_{n-1}\beta^{2^{n-1}})$, using the fact that we can express each number in $\{1,\ldots,N\}$ uniquely in its binary representation.  Further, we can even make this construction low-degree in all of the variables by considering $\beta(1+x_0\beta_0)(1+x_1\beta_1)\cdot(1+x_{n-1}\beta_{n-1})$ for new variables $\beta_0,\ldots,\beta_{n-1}$. This clearly embeds the previous construction so is still a hitting set, but is now the desired $\VP$-succinct generator.

\subsection{Algebraic Natural Proofs and Geometric Complexity Theory}\label{sec:GCT}

We now comment on the connection between algebraic natural proofs and the Geometric Complexity Theory (GCT) program of Mulmuley and Sohoni~\cite{GCT1}. This program posits a very well-motivated method for obtaining algebraic circuit lower bounds, drawing inspiration from algebraic geometry and representation theory. 

To begin, we briefly discuss some algebraic geometry, so that we now work over an algebraically closed field $\F$. Suppose we have a class of polynomials $\cC\subseteq\totdeg$, which we can thus think of as vectors in the space $\F^{\totdegN}$.  As we did before, we can look at classes of distinguisher polynomials $\cD\subseteq\F[c_1,\ldots,c_{\totdegN}]$ which take as inputs the vector of coefficients of a polynomial in $\totdeg$.  In particular, we wish to look at the class of distinguishers $D$ that vanish on all of $\cC$, that is $\cD=\{D: D(\coeff(f))=0,\ \forall f\in\cC\}$.  Thus, $\cD$ vanishes on $\cC$, but it also may vanish on other points.  The \emph{(Zariski) closure} of $\cC$, denoted $\overline{\cC}$, is simply all polynomials $f\in\totdeg$ which the distinguishers $\cD$ vanish on, that is $\overline{\cC}=\{f\in\totdeg: D(\coeff(f))=0,\ \forall D\in\cD\}$.  Clearly $\cC\subseteq\overline{\cC}$, but this is generally not an equality.  For example, consider the map $(x,y)\mapsto (x,xy)$.  It is easy to see that the image of this map is $\F^2\setminus (\{0\}\times(\F\setminus\{0\}))$, but the closure is all of $\F^2$ (for further examples related to algebraic complexity classes, see \cite{BIZ17}).

From the perspective of algebraic geometry, it is much more natural to study the closure $\overline{\cC}$ rather than the class $\cC$ itself.  And indeed, the algebraic natural proofs we define here necessarily give lower bounds for the closure $\overline{\cC}$ because the lower bound is proven using a distinguisher in $\cD$.  In fact, algebraic geometry shows that lower bounds for $\overline{\cC}$ \emph{necessarily} must use such distinguishers (though they may not have small circuit size).\footnote{It is unclear how much a difference this closure makes. For example, the exact relation between $\VP$ and $\overline{\VP}$ is unclear, see for example the work of Grochow, Mulmuley and Qiao~\cite{GrochowMQ16}. It is conceivable that the algebraic distinguisher approach tries to prove too much, that is, perhaps $\overline{\VP}=\VNP$. We refer again to \cite{BIZ17} for further discussion and examples which separate natural algebraic complexity classes from their closures.} Thus, we see that this distinguisher approach fits well into algebraic geometry and hence the GCT program.

Thus, the GCT approach fits into the algebraic natural proofs structure if one discards the (key) property of constructiveness. However, the GCT approach also uses more than just algebraic geometry and in particular relies on representation theory.  That is, the GCT program notes that polynomials naturally have symmetries through linear changes of variables $\vx\to A\vx$ for an invertible matrix $A$ and these symmetries act not only on the circuits $\cC$ being computed but also their distinguishers $\cD$.  One can thus then ask that the lower bounds methods respect these symmetries, and Grochow~\cite{Grochow15} showed that most lower bounds in the literature do obey the natural symmetries one would expect.  Although this is not exactly precise, a useful picture is that the goal of the GCT program is to use the symmetries of the distinguishers $\cD$ to narrow down the search for them.  

It is unclear to what extent constructivity plays a role in such arguments and as such the GCT program is not a-priori algebraically natural in the sense given here.  Indeed, if there is an algebraically natural proofs barrier then the distinguishers that vanish on $\VP$ must have super-polynomial complexity, so that then clearly GCT is not constructive.  This viewpoint demonstrates that our succinct hitting set constructions have relevance to GCT as they prove super-polynomial lower bounds for distinguishers that vanish on $\VP$ (also known as the defining equations), at least in the restricted models we consider:

\begin{corollary}
Let $\mathcal{T}$ be the set of defining equations for $\overline{\VP}$. For each of the models mentioned in \autoref{thm:succinct-results}, there exists a polynomial $P \in \mathcal{T}$ which requires super-polynomial size when computed in this model.
\end{corollary}

\subsection{Follow-up Work}

We end this section by briefly mentioning two related works that have appeared since the initial version of this paper was posted.

Efremenko, Garg, Oliveira and Wigderson~\cite{EGOW18} studied algebraic circuit lower bounds proved using subadditive complexity measures based on matrix rank.  Such \emph{rank-based} methods are often used in practice to prove lower bounds on restricted models of algebraic computation.  These lower bounds are algebraic, and also often fall in our framework of algebraic natural proofs (as often the corresponding matrices are polynomially-large in the relevant parameters, but this is not always true as seen in \cite{GKKS14}). The main results of \cite{EGOW18} are \emph{unconditional} barriers on proving tensor-rank lower bounds or Waring-rank lower bounds using rank-based methods.

Bl\"{a}ser, Ikenmeyer, Jindal and Lysikov~\cite{BIJL18} studied our notion of algebraic natural proofs in the context of a complexity measure they call \emph{border completion rank} of a affine linear matrix polynomial.  They establish (among other results) that there is an infinite family of linear matrices for which no algebraically natural proof can prove the matrices have high border completion rank, assuming that the polynomial hierarchy does not collase.  This underlying assumption is more widely believed than the conjecture that succinct hitting sets exist, but their conclusion does not rule out an algebraic natural proof for high border completion rank some some \emph{other} set of matrices.

\section{Preliminaries}\label{sec:prelim}

We use boldface letters to denote vectors, where the length of a vector is usually understood from the context. Vectors such as $\vecx, \vecy$ and so on denote vectors of variables, where as $\vecalpha, \vecbeta$ are used to denote vectors of scalars. Similar boldface letters are used to denote tuples of polynomials. As done in the introduction, we will express polynomials $f\in\F[x_1,\ldots,x_n]$ in their monomial basis $f(\vx)=\sum_\va c_\va \vx^\va$ and then the corresponding vector of coefficients $\coeff(f)=(c_\va)_\va$ can then be the input space to another polynomial $D\in\F[\{c_\va\}_\va]$.  The exact size of this coefficient vector will be clear from context, that is, whether $f$ is multilinear (so there are $\inddegNm=2^n$ coefficients) or whether $f$ is of total degree at most $d$ (so there are $\totdegN=\binom{n+d}{d}$ coefficients). Occasionally, we have a polynomial $f \in \F[\vecx, \vecy]$, and in that case we denote $\coeff_\vecx (f)$ the coefficient-vector of $f$ where we think of $f \in \left( \F[\vecy] \right) [\vecx]$, that is, the entries of the vector are now polynomials in $\vecy$.

\section{Universal Constructions of Pseudorandom Polynomials}\label{sec:universal}

In this section we detail \emph{universal circuits} and their applications to pseudorandom polynomials.  That is, a universal circuit for small computation is a polynomial $U(\vx,\vy)$ such that for any polynomial $f(\vx)$ computed by a small computation, there is some value $\vaa$ such that $f(\vx)=U(\vx,\vaa)$.  Intuitively, there should be such universal circuits due to various completeness results, such as the fact that the determinant is complete for algebraic branching programs (\cite{v79}) (and hence complete for $\VP$ under quasipolynomial-size reductions (\cite{vsbr83})).  One would then expect that \emph{if} there are pseudorandom polynomials then such universal circuits would also be pseudorandom.

Indeed, based on this intuition Aaronson and Drucker~\cite{aaronson-blog-post} gave a candidate construction of pseudorandom polynomials based on generic projections of the determinant, with the intention of exploiting the completeness of the determinant.  However, we note here that while it is plausible that this construction is in fact pseudorandom, it is insufficient for our requirement for universality, as we want the computed $f$ and the universal $U$ to live in the same space of polynomials.  That is, if $f$ is of low total-degree so that $f\in\totdeg$, then we want $U(\vx,\vaa)\in\totdeg$ for \emph{every} $\vaa$. This is because we want a collection of polynomials $\cC$ that is indistinguishable from generic polynomials in $\totdeg$.  If we start with such a collection and attempt to embed them into $U$ where $\deg_\vx U(\vx,\vy)=d'\gg d$, the resulting collection of polynomials necessarily lives in the larger space $\F[\vx]^{d'}$ and the indistinguishability property no longer clearly holds.  As a concrete example, suppose $f(\vx)$ is a ``generic'' polynomial in $\totdegs$.  Then the modified polynomial $f(\vx)+z$ still embeds $f$, yet it lives in $\F[\vx,z]^d$, where it is no longer generic as it is linear in $z$.  

Thus, we need a universal circuit construction that does not increase the degree of $\vx$.  For algebraic branching programs, the candidate of Aaronson and Drucker~\cite{aaronson-blog-post} is easy to fix by switching from the determinant to iterated matrix multiplication, which is also complete but due to efficient homogenization of branching programs (\cite{nis91}) can be universal without increasing degree. However, for full generality we want to be universal for circuits, that is, obtaining a polynomial $U$ complete for $\VP$ under polynomial-size reductions which also ensures the $\vx$-degree of $U$ matches that of $f$.  \Burgisser~\cite[Section 5.6]{bur00} first achieved results in this vein by using auxiliary variables to trace through a generic computation, using homogenization to ensure the $\vx$-degree is never larger than needed.  Unfortunately his construction yields exponentially large degree in $\vy$ so it is not sufficient here.  A construction with such low degree was given by Raz~\cite{Raz10a}. We now state this result.

\begin{theorem}[Raz~\cite{Raz10a}]\label{fact:universal}
	Let $\F$ be a field, and let $n,s\ge 1$ and $d\ge 0$. Then there is a $\poly(n,d,s)$-size algebraic circuit $U\in\F[x_1,\ldots,x_n,y_1,\ldots,y_r]$ with $r\le \poly(n,d,s)$ such that $U$ can be constructed in time $\poly(n,d,s)$, and
	\begin{itemize}
		\item $\deg_\vx U(\vx,\vy)\le d$
		\item $\deg_\vy U(\vx,\vy)\le \poly(d)$
		\item If $f\in\F[\vx]$ has $\deg_\vx f\le d$ and $f$ is computed by a size $s$ circuit, then there is some $\vaa\in\F^r$ such that $f(\vx)=U(\vx,\vaa)$.
	\end{itemize}
\end{theorem}

We briefly note that this construction also yields a universal circuit for homogeneous degree-$d$ computations (the space $\totdeghom$).  No such universal circuits are known for efficient multilinear computation (the space \inddegm), as circuits do not likely admit efficient multilinearization. In contrast, there is a universal circuit for the depth-3 set-multilinear formulas, which is the model that we use to construct our succinct hitting sets fooling restricted classes of computation.  However, we restrict attention to total degree $d$ polynomials as this is the cleanest setting.

We now use this universal circuit to convert from succinct hitting sets to succinct generators, as the standard conversion from hitting set to generator would ruin succinctness.

\begin{lemma}\label{res:universal:hit-to-gen}
	Let $\F$ be a field, and let $n,s\ge 1$ and $d\ge 0$. Let $\cD\subseteq\F[c_1,\ldots,c_\totdegN]$ be a class of polynomials in the coefficient vectors of $\totdeg$.  If there is an $s$-succinct hitting set for $\cD$ then there is a $\poly(n,d,s)$-succinct generator for $\cD$ computable by $\poly(n,d,s)$-size circuits.
\end{lemma}
\begin{proof}
	Let the $s$-succinct hitting set arise from the set of size-$s$ polynomials $\cC\subseteq\totdeg$.  Let $U\in\F[\vx,y_1,\ldots,y_r]$ be the universal circuit of \autoref{fact:universal}. Then for any $f\in\cC$ there is some $\vaa\in\F^r$ such that $f(\vx)=U(\vx,\vaa)$.  Thus, 
	\[
		\cC\subseteq U(\vx,\F^r)=\{U(\vx,\vaa): \vaa\in\F^r\}
		\;.
	\]
	Thus, we see that $U$ is indeed a generator for $\cD$ as it contains the hitting set $\cC$ in its image.  Further, $U(\vx,\vaa)\in\totdeg$ for all $\vaa\in\F^r$ by construction.  Finally $U(\vx,\vaa)$ is computable in size $\poly(n,d,s)$ for all $\vaa\in\F^r$ as $U(\vx,\vy)$ has such a circuit and the substitution $\vy\leftarrow\vaa$ does not increase circuit size.  It follows that $U$ is the desired succinct generator.
\end{proof}

As mentioned in the introduction, generators are more robust versions of hitting sets.  We now give another reason for this, by proving that succinct generators imply succinct hitting sets of \emph{small size}, by using the standard interpolation argument.

\begin{lemma}\label{res:universal:gen-to-hit}
	Let $\F$ be a field with $|\F|>\delta\Delta$, where $\Delta,\delta\ge0$. Let $n,s\ge 1$ and $d\ge 0$. Let $\cD\subseteq\F[c_1,\ldots,c_\totdegN]$ be a class of degree-$\Delta$ polynomials in the coefficient vectors of $\totdeg$.  Suppose that $G\in\F[\vx,y_1,\ldots,y_\ell]$ is a succinct generator computable in size-$s$ for $\cD$ where $\deg_\vy G\le \delta$. Then there is a $s$-succinct hitting set of size $(\delta\Delta+1)^\ell$.
\end{lemma}
\begin{proof}
	For any $D\in\cD$, we see that $D$ is non-zero iff $D(\coeff_\vx G(\vx,\vy))$ is non-zero as a polynomial in $\vy$.  In particular, $\deg_\vy D(\coeff_\vx G(\vx,\vy))\le \deg D\cdot \deg_\vy G\le \delta\Delta$.  Thus, as the field is large enough we can find a set $S\subseteq \F$ with $|S|\ge \delta\Delta+1$, so that by interpolation $D(\coeff_\vx G(\vx,\vy))$ is non-zero iff $D(\coeff_\vx G(\vx,\vaa))$ is non-zero for every $\vaa\in S^\ell$.  Thus, we see that $G(\vx,S^\ell)$ is the desired succinct hitting set as each $G(\vx,\vaa)$ has a size-$s$ circuit (as substitution does not increase circuit size) and $S^\ell$ has the correct size.
\end{proof}

In the usual range of parameters we would have $\Delta=\poly(N)$ and $\delta=\poly(n,s)$. Plugging this into the above connections, we see that \emph{any} (even infinite) succinct hitting set implies quasipolynomial-size hitting sets.

\begin{corollarywp}
	Let $\F$ be a field, and let $n\ge 1$. Consider polynomials in $\F[c_1,\ldots,c_N]$ where $N=\binom{2n}{n}$ so that $\F[c_1,\ldots,c_N]$ can be identified with the coefficients of polynomial sin $\totdeg$ with $d=n$.  If $\poly(N)$-size $\poly(N)$-degree circuits in $\F[c_1,\ldots,c_N]$ have $\poly(n)$-succinct hitting sets from $\F[\vx]^n$, then such circuits have an explicit $\poly(N)^{\polylog{N}}$-size hitting set.
\end{corollarywp}

\section{Succinct Hitting Sets via Rank Condensers}\label{subsec:vandermonde}

In this section, we construct succinct generators for restricted depth-3 formulas ($\SPS$ formulas), in particular, $\SPSk$ formulas (top-fan-in $k$) and depth-$3$ circuits with bounded transcendence degree. The constructions are based on a common tool which we dub \emph{succinct rank condenser}.

Gabizon and Raz~\cite{GR08}, in the context of studying deterministic extractors, studied how to pseudorandomly map $\F^n$ to $\F^r$ preserving vector spaces of dimension $r$ with high probability. In particular, they gave a $\poly(n)$-collection of linear maps $\cE=\{E:\F^n\to\F^r\}$ such that for any vector space $V\subseteq\F^n$ of dimension $r$ there was at least one map $E\in\cE$ such that the dimension of $V$ was preserved, that is, $\dim E(V)=\dim V=r$. Their construction was improved by Forbes-Shpilka~\cite{ForbesShpilka12}, and was called a \emph{rank condenser} in later works (\cite{FSS14,ForbesGuruswami15}) which further explored this concept.

Rank condensers have proven very useful in designing hitting sets as they can reduce $n$-variate polynomials to $r$-variate polynomials, and for us the Gabizon and Raz~\cite{GR08} construction suffices.  In particular, one defines the map $E\in\F[t]^{n\times r}$ with $E_{i,j}=t^{ij}$, with $t$ is a formal variable.  One can then obtain the desired collection $\cE$ by evaluating $E(t)$ at sufficiently many points in $t\in\F$.  However, it suffices for us to obtain generators, so we leave $t$ as a formal variable.

\begin{construction}[Succinct Rank Condenser]\label{con:succinct-rk-cond}	
	Let $n\ge r\ge 1$.  Define the polynomial $\Prc_{n,r}$ where $\Prc_{n,r}\in\F[x_1,\ldots,x_n,y_1,\ldots,y_r,t_0,t_1,\ldots,t_n]$ to be
	\[
		\Prc_{n,r}(\vx,\vy,\vt)=\sum_{j=1}^r y_jt_0^j\prod_{k=1}^n(1+x_kt_k^j)
		\;.
	\]
	Let $\Grc_{n,r}(\vy,\vt)$ be the polynomial map given by $\coeff_\vx(\Prc_{n,r})$ when taking $\Prc_{n,r}$ as a multilinear polynomial in $\vx$.
\end{construction}

We now analyze properties of \autoref{con:succinct-rk-cond}, in particular showing that it embeds the desired rank condenser of Gabizon and Raz~\cite{GR08}.

\begin{proposition}\label{cl:succinct-vdm-is-vdm}
	Assume the setup of \autoref{con:succinct-rk-cond}. Taking $N=2^n$, identify $[N]$ with $2^{[n]}$. Then for every $i\in[N]$,
	\[
	\left(\Grc_{n,r}(\vx,\vy,t,t^{2^0},t^{2^1}\cdots,t^{2^{n-1}})\right)_i = \sum_{j=1}^r y_j t^{ij}
	\]
\end{proposition}
\begin{proof}
	We can index the output coordinates of $\Grc_{n,r}$ with subsets $S\subseteq[n]$, so that an index $i\in[N]$ gets mapped to $S\subseteq[n]$ via its binary representation so that $i-1=\sum_{k\in S} 2^{k-1}$, and for a given $S\subseteq[n]$ denote the corresponding index $i_S$. Then,
	\begin{align*}
		\Prc_{n,r}(\vx,\vy,t,t^{2^0},t^{2^1}\cdots,t^{2^{n-1}})
		&=\sum_{j=1}^r y_jt^j\prod_{k=1}^n(1+x_k(t^{2^{k-1}})^j)\\
		&=\sum_{j=1}^r y_jt^j\sum_{S\subseteq[n]}\prod_{k\in S}x_k\cdot t^{j\cdot 2^{k-1}}\\
		&=\sum_{j=1}^r y_jt^j\sum_{S\subseteq[n]} t^{j\cdot \sum_{k \in S}2^{k-1}}\prod_{k\in S}x_k\\
		&=\sum_{j=1}^r y_jt^j\sum_{S\subseteq[n]} t^{j\cdot (i_S-1)}\prod_{k\in S}x_k
		\;.
	\end{align*}
	Thus, taking coefficients in $\vx$ exactly indexes $\sum_{j=1}^r y_j t^{ij}$ as required.
\end{proof}

We now observe that this generator is efficiently computable, and produces succinct hitting sets.

\begin{propositionwp}\label{fact:vdm-is-succinct}
	Assume the setup of \autoref{con:succinct-rk-cond}. The polynomial $\Prc_{n,r}(\vx,\vy,\vt)$ is computable by $\poly(n,r)$-size $\Sigma\Pi\Sigma\Pi$ circuits of $\poly(n,r)$-degree.  Further, for every fixing  $\vecy=\vecalpha\in\F^r$, $\vt=\vbb\in\F^{n+1}$, $\Prc_{n,r}(\vx,\vaa,\vbb)$ is computed by a $\SPS$ circuit of size $\poly(r, n)$.
\end{propositionwp}

\subsection{Depth-3 Formulas with Bounded Top-Fan-In}\label{subsec:spsk}

A \emph{$\SPSk$ formula} is a depth-$3$ formula of the form $\sum_{i=1}^{k} \prod_{j=1}^{d_i} \ell_{i,j}$, where $\ell_{i,j}$ are linear functions in $x_1, \ldots, x_N$. We denote the degree of the circuit by $d = \max_{i} d_i$.

The study of $\SPSk$ formulas was initiated by Dvir and Shpilka \cite{DS07}, who proved that in a \emph{simple} and \emph{minimal}\footnote{We omit the exact definitions here and refer the reader to \cite{sy} for a thorough discussion.}  $\SPSk$ circuit computing the zero polynomial, the rank of the linear functions $\set{\ell_{i,j}}$ is bounded by a number $R(k,d)$ that is independent of the number of variables $N$. The number $R(k,d)$ is called the \emph{rank bound} for this class of circuits. Karnin and Shpilka \cite{KarninS11} showed how to use the rank condenser construction of Gabizon and Raz in order to obtain a black-box identity testing algorithm, and improved rank bounds were later obtained (\cite{kaysar09, SaxenaS11, SaxenaS12, SaxenaS13}).

In this section, we construct a $\poly(n, k)$-$\SPS$ succinct hitting set for $\SPSk$ formulas, and we use the fact that the rank condenser generator, with a judicious choice of $r$, is a generator for $\SPSk$ formulas. The version we cite here is from the survey \cite{sy}.

\begin{fact}[Hitting set for $\SPSk$ Formulas]
\label{fact:spsk-gen}
Let $F(\vecX) \in \F[\vecX]$ be a polynomial computed by a $\SPSk$ degree $d$ formula.
Let $V : \F^r \to \F^N$ the linear transformation given by the $N \times r$ Vandermonde matrix $(V_t)_{ij} = t^{i \cdot j}$ for $1 \le i \le N$, $1 \le j \le r$. Then, for $r = R(k,d)+1$ where $R(k,d)= O(k^2 \log d)$ (over finite fields) or $R(k,d) = k^2$ (over infinite fields), it holds that $F \neq 0$ if and only if the $r$-variate polynomial $F \circ \left( V_t \cdot (y_1, \ldots, y_r)^T \right )$ is non-zero.
\end{fact}

Using \autoref{fact:spsk-gen} and the properties of \autoref{con:succinct-rk-cond}, we obtain the following two lemmas.
\begin{lemma}
\label{lem:spsk-succinct}
The polynomial map $\Grc_{n,R(k,d)} (\vecy,\vt) $ is $\poly(R(k,d), n)$-$\SPS$ succinct. In particular, the generator is $\poly(k, \log d, n)$-$\SPS$ succinct.
\end{lemma}

\begin{proof}
The first statement is immediate from \autoref{fact:vdm-is-succinct}.
The second statement follows using the rank bounds for $\SPSk$ formulas stated in \autoref{fact:spsk-gen}.
\end{proof}

\begin{lemma}
\label{lem:spsk-hits}
Let $F$ be computed by a $\SPSk$ formula. Then $F \circ \Grc_{n,R(k,d)} \not\equiv 0$.
\end{lemma}

\begin{proof}
Immediate from \autoref{cl:succinct-vdm-is-vdm} (making the appropriate substitution for $\vt$) and \autoref{fact:spsk-gen}.
\end{proof}

\begin{corollary}
$\Grc_{n,R(k,d)} (\vecy,\vt)$ is a $\poly(k, \log d, n)$-$\SPS$ succinct generator for the class of $\SPSk$ formulas.
\end{corollary}

\subsection{Depth-$3$ circuits of bounded transcendence degree}\label{subsec:trdeg}

We now generalize the results of \autoref{subsec:spsk} to obtain a succinct hitting set for the larger class of circuits with \emph{bounded transcendence degree}.

A set of polynomials $\set{F_1, \ldots, F_r} \subseteq \F[\vecX]$ is called \emph{algebraically independent} if for any non-zero polynomial $H \in \F[w_1, \ldots, w_r]$, $H(F_1, \ldots, F_r) \not \equiv 0$. Given a set of polynomials $\set{F_1, \ldots, F_\ell}$, the \emph{transcendence degree} of this set, denoted $\trdeg{F_1, \ldots, F_\ell}$, is the size of a maximal algebraically independent subset of $\set{F_1, \ldots, F_\ell}$.

Let $C(Y_1, \ldots, Y_M)$ be a circuit of polynomial degree, and for $i \in [m]$, let $T_i = \prod_{j=1}^{d} L_{i,j}$, where $L_{i,j} \in \F[X_1, \ldots, X_N]$ are linear functions. In \cite{ASSS16}, Agrawal et al.\ present a hitting set for polynomials of the form $F=C(T_1, \ldots, T_M)$, where $\trdeg{T_1, \ldots, T_m}$ is bounded by $k$ (the size of the hitting set is exponential in $k$). In this section we present a succinct version of their generator.

\begin{lemma}[Generator for circuits of transcendence degree $k$, \cite{ASSS16}, and see also the presentation in Chapter 4 of \cite{rp}]
\label{lem:trdeg-k-generator}
Suppose $\F$ is a field such that $\Char(\F) = 0$ or $\Char(\F) \ge d^k$. Then the map $\Psi : \F[\vecX] \to \F[y_1, \ldots, y_k, t, z_1, \ldots, z_k, s]$, given by
\[
X_i \mapsto \sum_{j=1}^{k+1} z_j s^{ij} + \sum_{j=1}^k y_j t^{ij}
\]
for every $i \in [N]$, is a generator for the class of polynomials $F \in \F[\vecX]$ expressible  as $C(T_1, \ldots, T_M)$, where the $T_i$'s are products of linear functions and $\trdeg{T_1, \ldots, T_m} \le k$.
\end{lemma}

It remains to be noted that we can construct the map $\Psi$ succinctly.

\begin{theorem}
\label{thm:succinct-trdeg}
Suppose $\F$ is a field such that $\Char(F) = 0$ or $\Char(F) \ge d^k$. Then there exists a $\poly(k, n)$-$\SPS$ succinct generator for the class of polynomials that can be represented as $C(T_1, \ldots, T_M)$ with $C$ being a $\poly(N)$ degree circuit, each $T_i$ is a product of $d$ linear functions and $\trdeg{T_1, \ldots, T_M} \le k$.
\end{theorem}

\begin{proof}
Observe that $\Psi$ from \autoref{lem:trdeg-k-generator} can be represented as $\coeff_{\vecx} (P(\vecy, \vecz, \vs, \vt))$, where
\[
P(\vecx, \vecy, \vecz, \vs, \vt) = \Prc_{n,k+1} (\vecx, \vecz, \vs) + \Prc_{n,k} (\vecx, \vecy, \vt).
\]
The succinctness follows from \autoref{fact:vdm-is-succinct}, and from observing that $\poly(k,n)$-$\SPS$ circuits are closed under addition.
\end{proof}

\section{Succinct Hitting Sets via the Shpilka-Volkovich Generator}\label{sec:sv}

The Shpilka-Volkovich Generator (SV Generator, henceforth, and see \cite{SV15}) is a polynomial map $\cG (y_1, \ldots, y_k, z_1, \ldots, z_k) : \F^{2k} \to \F^{N}$ that satisfies the property that for every $T \subseteq [N]$ such that $|T| \le k$, we can set $z_1, \ldots, z_k$ to values $\alpha_{i_1}, \ldots, \alpha_{i_k}$ such that the $\vecy$ variables are mapped to the locations indexed by $T$, and the other coordinates of the polynomial map are zeroed out.  This property turns out to be immensely useful in constructing hitting sets for various classes. Hence, we begin by constructing a succinct analog of this generator, and then use it to obtain succinct hitting sets in cases where the SV generator is applicable.

\begin{construction}[Succinct SV Generator]
\label{construction:ssv}
Let $n \in \N$ and $N=2^n$. Define
\[
P(z_1, \ldots, z_n, x_1, \ldots, x_n) = \prod_{i=1}^n (z_i \cdot x_i  + (1-z_i)),
\]
and
\[
Q^{\SSV}_{n,k}(y_1, \ldots, y_k ,z_{1,1}, \ldots, z_{1,n}, \ldots, z_{k,1}, \ldots, z_{k,n}, x_1, \ldots, x_n) = \sum_{i \in [k]} y_i \cdot P(\vecz_i, \vecx), 
\]
where $\vecz_i = (z_{i,1}, \ldots, z_{i,n})$. Finally, let
\[
\cG_{n,k}^{\SSV} (y_1, \ldots, y_k, \vecz_1, \ldots, \vecz_k) = \coeff_{\vecx} ( Q_{n,k}^{\SSV}(\vecy,\vecz,\vecx) ). \qedhere
\]
\end{construction}

We begin by stating some immediate facts regarding \autoref{construction:ssv}.

\begin{fact}[Succinctness]
\label{fact:ssv-succinct}
For every setting $\vecy=\vecalpha, \vecz=\vecbeta$, the polynomial $Q_{n,k}^{\SSV}$ is computed by a multilinear $\Sigma\Pi\Sigma$ circuit of size $\poly(n,k)$.
\end{fact}

\begin{fact}[Additivity]
\label{fact:ssv-additive}
The succinct SV-generator is additive in $\vecy, \vecz$, in the sense that as polynomials, we have the equality
\[
Q^{\SSV}_{n, k_1} (\vecy_1, \vecz_1, \vecx) + Q^{\SSV}_{n, k_2} (\vecy_2, \vecz_2, \vecx) = Q^{\SSV}_{n, k_1+k_2} (\vecy', \vecz', \vecx),
\]
where $\vecy' = (\vecy_1, \vecy_2)$ and $\vecz' = (\vecz_1, \vecz_2)$. In particular, since the mapping from a polynomial to the coefficients vector is linear, as polynomial maps we get the equality
\[
\cG_{n,k_1}^{\SSV} (\vecy_1, \vecz_1) + \cG_{n,k_2}^{\SSV} (\vecy_2, \vecz_2) = \cG_{n, k_1 + k_2}^{\SSV} (\vecy', \vecz').
\]
\end{fact}

The usefulness of the generator comes from the following property, which is, in some sense, the algebraic analog of $k$-wise independence.

\begin{lemma}
\label{lem:SSV-indep-prop}
For every $T \subseteq [N]$ such that $|T| \le k$, there is a fixing of the $\vecz$ variables, and possibly of some of the $\vecy$ variables, such that in the mapping $\cG_{n,k}^{\SSV}$, $|T|$ distinct $\vecy$ variables are planted in the coordinates corresponding to $T$, while the rest of the entries are zeroed out.
\end{lemma}

\begin{proof}
As before, it is convenient to think of a subset of the $N$ coordinates as family of subsets of $[n]$.

Since $\cG_{n,k}^{\SSV}$ is given by the coefficients map of the polynomial $Q_{n,k}^{\SSV}(\vecy, \vecz, \vecx)$, an equivalent form of interpreting the statement of the lemma is that we want to fix the $\vecz$ variables such that distinct $\vecy$ variables become the coefficients of the monomials $\vecx_S$, for $S \in T$, and the coefficients of all monomials not in $T$ are zero.

Suppose first $|T|=k$ and denote $T = \set{S_1, \ldots, S_k}$. For every $j \in [k]$ set $\vecz_j = \vecone_{S_j}$, the characteristic vector of the set $S_j \subseteq [n]$. That is, $z_{j,i}=1$ if $i \in S_j$, and $0$ otherwise.

Observe that, in the notation of \autoref{construction:ssv}, we have that
\[
P(\vecone_{S_j}, x_1, \ldots, x_n) = \prod_{i=1}^n ((\vecone_{S_j})_i \cdot x_i + (1-(\vecone_{S_j})_i)) = \prod_{i : (\vecone_{S_j})_i = 1} x_i = \prod_{i \in S_j} x_i = \vecx_{S_j}.
\]
Therefore, we get that
\[
Q_{n,k}^{\SSV} (y_1, \ldots, y_k, \vecone_{S_1}, \ldots, \vecone_{S_k}, \vecx) = \sum_{i \in [k]} y_i \vecx_{S_i},
\]
as we wanted.

If $|T|=k'<k$, we can arbitrarily extend $T$ so a set $T'$ of size exactly $k$, and then set some $\vecy$ variables to zero, in order to zero out the relevant $k-k'$ entries in the polynomial map. 
\end{proof}

Suppose we aim to hit a polynomial $F \in \F[\vecX]$ of degree $d$, and we are given the information that $F$ contains a non-zero monomial with at most $k$ variables. Assuming $k$ is small, a natural algorithm in that case is to ``guess'' the $m \le k$ variables in the small support monomial, zero out all the remaining variables, and then do use the trivial derandomization, using the Schwartz-Zippel-DeMillo-Lipton Lemma, with respect to the remaining $k$-variate polynomial, for a cost of $(d+1)^k$ many evaluations. This is exactly what the SV generator enables us to do, since we can set the $\vecz$ variables in a way that the $k$ $\vecy$ variables will contain those that appear in the small support monomial, and thus, since after fixing the $\vecz$ variables the polynomial remains non-zero, it follows that it is non-zero even without fixing the $\vecz$ variables.
In this subsection we use this simple idea to construct succinct hitting sets for several classes of circuits. A small caveat is that usually we are \emph{not} guaranteed our target polynomial has a small support monomial, but we can prove that this is the case after a proper shift of the $N$ variables (one of course also has to represent the shift succinctly in $n$).

A similar notion was used by Agrawal, Saha and Saxena \cite{ASS13} to show that certain classes
of polynomials simplify under shifts in a way which is helpful for designing PIT algorithms. In their case, the shift is by a vector of polynomials in a set of formal variables $\vec{t}$, whereas in our case the shifts are much simpler: for our applications we only need to shift by the constant vector $\vecone$. 

\begin{construction}[Succinct Hitting Set for classes with small support monomials after shifts by $\vecone$]
\label{con:hitting-set-small-support}
Let $k, n \in \N$ and $N=2^n$. Define the \emph{shifted succinct SV polynomial} to be
\[
	Q_{n,k}^{\SSSV} (y_1, \ldots, y_k, z_{1,1}, \ldots, z_{1,n}, \ldots, z_{k,1}, \ldots, z_{k,n}, x_1, \ldots, x_n) = Q_{n,k}^{\SSV} (\vecy, \vecz, \vecx) + \prod_{i=1}^n (x_i + 1),
\]
and the  \emph{shifted succinct SV generator} as
\[
\cG_{n,k}^{\SSSV} (\vecy, \vecz) = \coeff_{\vecx} (Q_{n,k}^{\SSSV} ). \qedhere
\]
\end{construction}

We record the following simple fact, which follows from \autoref{fact:ssv-succinct}, and from the fact that $\coeff(\prod_{i=1}^n (x_i + 1)) = \vecone$.

\begin{fact}
\label{fact:sssv-succinct}
The generator $\cG_{n,k}^{\SSSV}$ is $\poly(k,n)$-$\SPS$ succinct, and as polynomial maps, we have the equality
\[
\cG_{n,k}^{\SSSV} (\vecy, \vecz) = \cG_{n,k}^{\SSV} + \vecone.
\]
\end{fact}

The following lemma shows how the shifted SV generator is useful for hitting classes of polynomials that have small support monomials after shifting by $\vecone$.

\begin{lemma}
\label{lem:small-support-hit}
Let $\cC$ be a class such that for all $f \in \cC$, $F(\vecX + 1)$ contains a monomial of support at most $k$. Then if $F \not\equiv 0$, $F \circ \cG_{n,k}^{\SSSV} (\vecy, \vecz) \not\equiv 0$.
\end{lemma}

\begin{proof}
Let $F(\vecX)$ be a non-zero polynomial from $\cC$, and let $G(\vecX)=f(\vecX + \vecone)$. By the assumption, $G$ is a non-zero polynomial that contains a monomial $M$ of support at most $k$. Let $S=\set{X_{i_1}, \ldots, X_{i_{k'}}}$ (where possibly $k' < k$) denote the subset containing exactly the variables in $M$, and consider $G \circ \cG_{n,k}^{\SSV} (\vecy, \vecz)$. By \autoref{lem:SSV-indep-prop}, we can set the $\vecz$ variables to $\vecalpha$ and possibly some of the $\vecy$ variables to $\vecbeta$ such that $y_1, \ldots, y_{k'}$ are mapped to $X_{i_1}, \ldots, X_{i_{k'}}$, and all the other variables are mapped to 0. Under this setting $g \circ \cG_{n,k}^{\SSV} (y_1, \ldots, y_{k'}, \vecalpha, \vecbeta) \not\equiv 0$, since the monomial $M$ is mapped to a monomial in $y_1, \ldots, y_{k'}$ which cannot be canceled out. Hence, $G \circ \cG_{n,k}^{\SSV} (\vecy, \vecz) \not \equiv 0$.

Finally, observe that $F \circ \cG_{n,k}^{\SSSV} (\vecy, \vecz) = G \circ \cG_{n,k}^{\SSV}(\vecy, \vecz)$.
\end{proof}

\subsection{Sparse Polynomials}\label{subsec:sparse}
In this section we give a $\poly(n)$-$\SPS$ succinct hitting set for the class of $\poly(N)$-sparse polynomials, i.e., polynomial size $\Sigma\Pi$ circuits. We note that an $s$-sparse polynomial $f$ can also be computed by a commutative roABP of width $s$, so in a sense, the results in this section are subsumed by those in \autoref{subsec:comm-roabp}. However, the argument made here is simpler and slightly more general since it applies to any class that has (possibly after shifting) small support monomials (see \autoref{subsec:smesp}).

We begin by recording the following fact.

\begin{lemma}[\cite{Forbes15, GKST16}]
\label{lem:small-support-sparse}
Let $F \in \F[X_1, \ldots, X_N]$ by a polynomial with at most $s$ monomials, and $\vecalpha \in \F^N$ be a full support vector, that is, for all $i \in [N]$, $\alpha_i \neq 0$. Then the polynomial $F(\vecX + \vecalpha)$ has a monomial of support at most $\log s$.
\end{lemma}

This lemma appears in \cite{Forbes15} and \cite{GKST16} with two very different proofs. For completeness, we provide yet a third proof, which we find to be more elementary. The proof relies upon the following easy lemma, due to Oliveira, which can be proved by induction on $N$ (see, e.g., \cite{FSTW16}).

\begin{lemma}[see Proposition 6.14 in \cite{FSTW16}]
\label{lem:multilinear-prod-sparsity}
Let $F \in \F[X_1, \ldots, X_N]$ be a multilinear polynomial with at most $s$ monomials, and $G \in \F[X_1, \ldots, X_N]$ be any non-zero polynomial. Then $F \cdot G$ has at most $s$ monomials. 
\end{lemma}

We now give our proof for \autoref{lem:small-support-sparse}.

\begin{proof}[Proof of \autoref{lem:small-support-sparse}]
Suppose, towards contradiction, that the minimal monomial in $G(\vecX) := F(\vecX + \vecalpha)$ has $\ell \ge \log s + 1$ variables. Further suppose, without loss of generality, these are $X_1, \ldots, X_\ell$. Consider now $G(X_1, X_2, \ldots, X_\ell, 0, \ldots, 0)$. By assumption, this is a non-zero polynomial which is divisible by the monomial $X_1 X_2 \cdots X_\ell$. It follows that 
\[
F(X_1, \ldots, X_\ell, \alpha_{\ell+1}, \ldots, \alpha_n)  = G(X_1 -\alpha_1 , \ldots, X_\ell - \alpha_\ell, 0, \ldots, 0) =  
\left( \prod_{i=1}^\ell (X_i - \alpha_i) \right) \cdot H(X_1, \ldots, X_\ell),
\]
for some non-zero $H$.

Since $\prod_{i=1}^\ell (X_i - \alpha_i)$ is multilinear of sparsity $2^\ell > s$, it follows from \autoref{lem:multilinear-prod-sparsity} that the sparsity of $F(X_1, \ldots, X_\ell, \alpha_{\ell+1}, \ldots, \alpha_n)$ is also greater than $s$, which contradicts the assumption on $F$, as the sparsity can only decrease when fixing variables.
\end{proof}

\autoref{lem:small-support-sparse}, along with \autoref{lem:small-support-hit} and \autoref{fact:sssv-succinct} immediately imply that the shifted succinct SV generator hits sparse polynomials.

\begin{corollarywp}
\label{cor:hitting-set-sparse}
The generator $\cG_{n, \log s}^{\SSSV}$ from \autoref{con:hitting-set-small-support} is a $\poly(\log s, n)$-$\SPS$ succinct generator for the class of $s$-sparse polynomials $F \in \F[X_1, \ldots, X_N]$.
\end{corollarywp}

\subsection{Sums of Powers of Low Degree Polynomials}\label{subsec:smesp}

We now mention another class that, after a suitable shifting, has small support monomials.

\begin{definition}[$\SMESP{t}$ formulas]
A polynomial $F(\vecX) \in \F[\vecX]$ is computed by a $\SMESP{t}$ formula if
\[
F(\vecX) = \sum_{i=1}^{s} \vecX^{\veca_i} F_i (\vecX)^{d_i},
\]
where $\deg F_i \le t$ for all $i \in [s]$, and $\vecX^{\veca_i} = \prod_{j=1}^{N} X_i^{a_{i,j}}$ is a monomial.
\end{definition}

The following was proved in \cite{Forbes15}.

\begin{lemma}
\label{lem:SMESP-shift}
Suppose $F[\vecX]$ is computed by a $\SMESP{O(1)}$ formula of top fan-in $s$, and let $\vecalpha$ be a full-support vector. Then it holds that $F(\vecX + \vecalpha)$ has a monomial of support at most $O(\log s)$. 
\end{lemma}

It follows that a similar construction to the one which we used to succinctly hit sparse polynomials also works in this case.

\begin{theorem}
There exists a $\poly(\log s, n)$-$\SPS$ succinct generator for the class of $\SMESP{O(1)}$ formulas of top fan-in $s$.
\end{theorem}

\begin{proof}
Let $C \cdot \log s$ the sparsity bound in \autoref{lem:SMESP-shift} and consider the generator $\cG_{n, C\log s}^{\SSSV}$ from \autoref{con:hitting-set-small-support}. By \autoref{fact:sssv-succinct}, this generator is $\poly(\log s, n)$-$\Sigma\Pi\Sigma$ succinct. By \autoref{lem:SMESP-shift} and \autoref{lem:small-support-hit}, it follows that $\cG_{n, C\log s}^{\SSSV}$ hits this class.
\end{proof}

\subsection{Commutative Read-Once Oblivious Algebraic Branching Programs}\label{subsec:comm-roabp}

In this section, we construct a $\poly(\log w, n)$-$\Sigma\Pi\Sigma$ succinct hitting sets for the class of $N$-variate polynomials computed by a width $w$ \emph{commutative} read-once oblivious algebraic branching programs.

A read-once oblivious algebraic branching program (roABP) is a directed, acyclic graph with the following properties:

\begin{itemize}
\item The vertices are partitioned into $N+1$ layers $V_0, \ldots, V_N$, such that $V_0 = \set{s}$ and $V_N = \set{t}$. $s$ is called the \emph{source node}, and $t$ the \emph{sink node}.
\item Each edge goes from $V_{i-1}$ to $V_{i}$ for some $i \in [N]$.
\item There exists a permutation $\sigma : [N] \to [N]$ such that all edges in layer $i$ are labeled by a univariate polynomial in $X_{\sigma(i)}$ of degree at most $d$.
\end{itemize}

We say that each $s \to t$ path in the ABP computes the product of its edge labels, and the roABP computes the sum over the polynomials computed by all $s \to t$ paths. The width of the roABP is defined to be $\max_i |V_i|$.

Equivalently, $F$ is computed by a roABP in variable order $\sigma$ if there exist $N$ matrices $M_1, \ldots, M_N$ of size $r \times r$ such that each entry in $M_i$ is a univariate, degree $d$ polynomial in $X_{\sigma (i)}$, and $F = \left(  \prod_{i=1}^N M_i (X_{\sigma(i)}) \right) _{1,1}$. 

In general, it is possible for a polynomial to be computed by a small width roABP in a certain variable order, but to require a much larger width if the roABP is in a different variable order. A polynomial $f \in \F[\vecX]$ is computed by a  width $w$ commutative roABP, if it is computable by a width $w$ ABP in \emph{every} variable order.

Forbes, Saptharishi and Shpilka (\cite{FSS14}, Corollary 4.3) showed that in order to hit width-$w$ commutative roABPs, it is enough to take the SV generator with $k=O(\log w)$. \footnote{An improved construction for this model, with respect to the size of the hitting set, was given by Gurjar, Korwar and Saxena~\cite{GKS17}. Their construction, however, uses ingredients which we do not know how to make succinct; see \autoref{sec:discussion} for further discussion)}

We follow the proof strategy of \cite{FSS14} in order to show that the succinct SV generator hits commutative roABPs as well. The following definitions and the theorem following them are borrowed from \cite{FSS14}.

\begin{definition}
\label{def:indep-mon-map}
Let $g : \F^m \times \F^{m'} \to \F^N$ be a polynomial map. $g$ is said to be an \emph{individual degree $d$, $\ell$-wise independent monomial map} if for every $S \subseteq [N]$ of size at most $\ell$, there is $\vecalpha \in \F^{m'}$ such that the polynomials $\setdef{g(\vect,\vecalpha)^{\veca}}{\supp(\veca) \subseteq S, \max_i a_i \le d}$ are non-zero and distinct monomials in $\vect$, where we define
\[
g(\vect,\vecalpha)^{\veca} = \prod_{i=1}^N (g(\vect, \vecalpha)_i^{a_i}). \qedhere
\]
\end{definition}

\begin{definition}[see also \cite{ASS13}]
Let $\vec{F}[\vecX] \in \F[\vecX]^{r}$ be a vector of polynomials. We say that $\vec{F}$ has support-$k$ rank concentration at $\vecv$, if the derivatives of $\vec{F}$ with respect to all monomials of support at most $k$ at $\vecv$ span all the derivatives of $\vec{F}$ at $\vecv$. That is, if
\[
\Span\left\{ \partial_{\vecX^{\veca}} (\vec{F}) (\vecv) \right\}_{\left\{ \veca : \left|\supp(\veca)\right| \le k \right\} }
=
\Span\left\{ \partial_{\vecX^{\veca}} (\vec{F}) (\vecv) \right\}_{\veca } \qedhere
\]
\end{definition}

\begin{lemma}[\cite{FSS14}, Theorem 4.1]
\label{lem:mon-indep-concentrates}
Let $\vec{F}[\vecX] \in \F[\vecX]^{w \times w}$ be of individual degree $d$ and computed by a commutative roABP of width $w$. Let $g(\vect, \vecs)$ be an individual degree $d$, $(\log (w^2) + 1)$-wise independent monomial map. Then $\vec{F}(\vecX)$ has support-$\log(w^2)$ rank concentration at $g(\vect, \vecs)$ over $\F(\vect, \vecs)$.
\end{lemma}

The succinct SV generator, like the SV generator, is a $k$-wise independent monomial map for and degree $d$.

\begin{lemma}
\label{lem:SSV-mon-indep}
The polynomial map $\cG_{n,k}^{\SSV} (\vecy, \vecz) $ of \autoref{construction:ssv} is an individual degree $d$, $k$-wise independent monomial map for every $d$.
\end{lemma}

\begin{proof}
By \autoref{lem:SSV-indep-prop}, for any set $S$ of coordinates of size at most $k$ we can set the $\vecz$ variables such that each coordinate in $S$ contains a distinct $y$ variable. Then, it is also clear that all individual degree up to $d$ monomials of this map are distinct, for any choice of $d$
\end{proof}

The final ingredient (also from \cite{FSS14}) is the following lemma, which shows how to obtain hitting sets from rank concentration.

\begin{lemma}[\cite{FSS14}, Corollary 3.5]
\label{lem:concentration-to-hitting}
Let $\vec{F} \in \F[\vecX]^{r \times r}$ be a matrix of polynomials that is support-$k$ rank concentrated at $\vecalpha \in \F^N$, and let $G(\vecX) = \vec{F}_{1,1}$. Then $G(\vecX) \not\equiv 0$ if and only if $G \circ (\cG_{n,k}^{\SSV} + \vecalpha) \not\equiv 0$.
\end{lemma}

We remark that although \cite{FSS14} phrase this lemma for their construction of the SV generator, the proof goes through verbatim using the properties of $\cG_{n,k}^{\SSV}$ as explained in the proof of \autoref{lem:SSV-mon-indep}, and does not depend on the specific implementation.

We now prove that $\cG_{n, 4\log w + 1}$ hits $N$-variate polynomials that are computed by width $w$ commutative roABPs.

\begin{theorem}
\label{thm:SSV-hits-comm-roABPs}
Let $|\F| > nd$, and let $F(\vecX) \in \F[\vecX]$ be an $N$-variate polynomial of individual degree at most $d$, and computed by a width $w$ commutative roABP. Then, $F \not\equiv 0$ if and only if $F \circ \cG_{n, 1+ 4 \log w}^{\SSV} \not\equiv 0$.
\end{theorem}

\begin{proof}
By definition, $F$ is the $(1,1)$ entry of a matrix polynomial $\vec{F}(\vecX) \in \F[\vecX]^{w \times w}$, with $\vec{F}$ being computed by a width-$w$ commutative roABP. By \autoref{lem:SSV-mon-indep} and \autoref{lem:mon-indep-concentrates}, we get that the polynomial $\vec{F} \circ \cG_{n,2\log w+1}^{\SSV}$ is support-$\log(w^2)$ rank concentrated. By \autoref{lem:concentration-to-hitting}, we deduce that $F \not \equiv 0$ if and only if $F \circ ( \cG^{\SSV}_{n,2\log w+1} (\vecy_1, \vecz_1) + \cG^{\SSV}_{n,2\log w}(\vecy_2, \vecz_2) ) \not \equiv 0$, where $\vecy_1, \vecy_2, \vecz_1, \vecz_2$ are disjoint sets of variables. By the additivity property (\autoref{fact:ssv-additive}), it holds that  $F \not\equiv 0$ if and only if $F \circ \cG_{n, 1+ 4 \log w}^{\SSV} (\vecy, \vecz) \not\equiv 0$ for $\vecy=(\vecy_1, \vecy_2)$ and $\vecz = (\vecz_1, \vecz_2)$.
\end{proof}

\begin{corollary}
There exists a $\poly(n, \log(w))$-$\SPS$ succinct generator for the class of polynomials computed by width-$w$ commutative roABPs.
\end{corollary}

\subsection{Depth-$D$ Occur-$k$ Formulas}\label{sec:occurk}

The following model was considered in \cite{ASSS16}.

\begin{definition}
\label{def:depth-D-occur-k}
An \emph{occur-$k$ formula} is a directed tree, with internal nodes labeled either by $+$ or $\powerproduct$ (a power-product gate). The edges entering a $\powerproduct$ gate are labeled by integers $e_1, \ldots, e_m$, and on inputs $g_1, \ldots, g_m$, the gate computes $g_1^{e_1} \cdots g_m^{e_m}$. The leaves of tree are depth-$2$ formulas which compute sparse polynomials, such that every variable $X_i$ occur in at most $k$ of them. 

The \emph{size} of an occur-$k$ formula is the sum over the sizes of its gates, where
\begin{enumerate}
\item The size of a `$+$' gate is $1$,
\item The size of a `$\powerproduct$' gate is the sum $e_1 + \cdots + e_m$ of the labels of its incoming edges, and
\item The size of a leaf node is the size of the depth-$2$ formula it is computing.
\end{enumerate}

The \emph{depth} of an occur-$k$ formula is the number of layers of $+$ and $\powerproduct$ gates, plus 2, to account for the sparse formulas at the leaves.
\end{definition}

Agrawal et al.\ (\cite{ASSS16}) constructed a hitting set for this class, which combines both the rank condenser construction (\autoref{con:succinct-rk-cond}) and a generator for sparse polynomials. While the original construction uses the Klivans-Spielman generator (\cite{ks01}), it is possible to make the hitting set succinct while using our version of the shifted succinct Shpilka-Volkovich generator.

We now present the succinct generator of depth-$D$ occur-$k$ formulas.

\begin{construction}
\label{con:depth-D-occur-k}
Let $D, k, n, s \in \mathbb{N}$. Denote $R=(2k)^{2D \cdot 2^D}$. For every $\ell \in [D-2]$, let $\vecy_\ell = (y_{\ell,1}, \ldots, y_{\ell, R}$) denote a tuple of $R$ variables. We define the polynomial
\[
P^{\ASSS} (\vecx, \vecy_1, \ldots, \vecy_{D-2}, \vt_1, \ldots, \vt_\ell, \vecu, \vecv) = 
\sum_{\ell=1}^{D-2} \Grc_{n,R} (x_1, \ldots, x_n, \vecy_{\ell}, \vt_\ell)
 + Q_{n,R \log s + R \log R}^{\SSSV} (\vecx, \vecu, \vecv),
\]
and the generator
\[
\cG^{\ASSS}(\vecy_1, \ldots, \vecy_{D-2}, \vt_1, \ldots, \vt_\ell, \vecu, \vecv) = \coeff_{\vecx} (P^{\ASSS}). \qedhere
\]

\end{construction}

In our setting, we think of $k,D=O(1)$, which immediately implies:

\begin{fact}
\label{fact:ASSS-succinct}
For $k,D=O(1)$, the generator of \autoref{con:depth-D-occur-k} is $\poly(\log s, n)$-$\SPS$ succinct.
\end{fact}

We now quote (a variant of) a theorem proved by Agrawal et al., which shows that \autoref{con:depth-D-occur-k} is a generator for depth-$D$ occur-$k$ formulas.

\begin{theorem}[\cite{ASSS16}, and see also the presentation in Chapter 4 of \cite{rp}]
\label{thm:asss-depth-D-occur-k}
Suppose $\Phi(\vecw) : \F^{m} \to \F^N$ is a map such that for any polynomial $F(\vecX) \in \F[\vecX]$ of sparsity at most $R! \cdot s^R$, $F \circ \Phi \not\equiv 0$. Then there exist integers $r_1, \ldots, r_{D-2} \in [R]$, for $R=(2k)^{2D \cdot 2^D}$ such that the map
\begin{equation}\label{eq:ASSS-depth-D-occur-k}
\Psi : X_i \mapsto \sum_{\ell = 1}^{D-2} \left( \sum_{j=1}^{r_\ell} y_{j,\ell} t_{\ell}^{i j} \right) + \Phi(\vecw)
\end{equation}
is a generator for polynomials computed by depth-$D$ occur-$k$ formulas of size $s$ assuming $\Char(\F) = 0$ or $\Char(\F) > s^R$.
\end{theorem}

As a corollary, we obtain that \autoref{con:depth-D-occur-k} is a succinct generator for this class.

\begin{corollary}
\label{cor:succinct-depth-D-occur-k}
For $D, k = O(1)$, \autoref{con:depth-D-occur-k} is a $\poly(\log s, n)$-$\SPS$ succinct generator for the class of polynomials computed by size-$s$ depth-$D$ occur-$k$ formulas.
\end{corollary}

\begin{proof}
The succinctness claim follows from \autoref{fact:ASSS-succinct}.

For the hitting property, observe that by \autoref{cor:hitting-set-sparse}, the polynomial map $\cG_{n,m}^{\SSSV} (\vecu, \vecv)$ satisfies the properties required from $\Phi$ in \autoref{thm:asss-depth-D-occur-k} for $m=R \log s + R \log R$, and by \autoref{cl:succinct-vdm-is-vdm}, the generator
\[
\sum_{\ell=1}^{D-2} \Grc_{n,R} (\vecy_{\ell}, \vt_\ell)
\]
maps every $X_i$ to the polynomial $\sum_{\ell = 1}^{D-2} \left( \sum_{j=1}^{R} y_{j,l} t_{\ell}^{i j} \right)$ after using the substitutions in $\vt_\ell$ to a new variable $t_\ell$ as given in \autoref{cl:succinct-vdm-is-vdm}. Since $r_\ell \le R$, the polynomial in \eqref{eq:ASSS-depth-D-occur-k} is a projection of 
\autoref{con:depth-D-occur-k}, by possibly restricting excess $\vecy_\ell$ variables to 0.

The claim now follows from \autoref{thm:asss-depth-D-occur-k}. 
\end{proof}

\section{Succinct Hitting Sets for Circuits of Sparsely Small Transcendence Degree}\label{subsec:gen-sparse-low-trans}

Another model, which was considered in \cite{BMS13}, is that of circuits of the form $C(F_1, \ldots, F_m)$ where the $f_i$'s are polynomials of maximal sparsity $s$, $\trdeg{F_1, \ldots, F_m} = r$ and $C$ is an arbitrary circuit. It is possible to simplify the construction using ideas from \cite{ASSS16}, and thus we cite some of the definitions and the lemmas in the latter paper. Since we do not provide full proofs and do not discuss the full background, our terminology is slightly different at certain points.

We begin with the definition of the Jacobian matrix.

\begin{definition}
\label{def:jacobian}
Let $\vec{F} = \set{F_1(\vecX), \ldots, F_m(\vecX)} \subseteq \F[\vecX]$ be a set of $N$-variate polynomials. The Jacobian matrix of $\vec{F}$, denoted $\Jac_{\vecX} (\vec{F})$, is an $m \times N$ matrix such that $\Jac_{\vecX} (\vec{F}) _ {i,j} = \partial F_i / \partial X_j$.
\end{definition}

The rank of the Jacobian matrix captures the transcendence degree of $\vec{F}$, assuming the characteristic is $0$ or large enough.

\begin{fact}[\cite{BMS13}]
\label{fact:jacobian-rank}
Let $\vec{F} = \set{F_1(\vecX), \ldots, F_m(\vecX)} \subseteq \F[\vecX]$ be a set of $N$-variate polynomials over $\F$ of degree at most $d$, such that $\trdeg{F_1, \ldots, F_m} = r$. If $\Char(\F) = 0$ or $\Char(\F) \ge d^r$, then $\rank_{\F(\vecX)} \Jac_{\vecX} (\vec{F}) = r$.
\end{fact}

This fact shows that a map that preserves the rank of the Jacobian also preserves the transcendence degree of the $f_i$'s, a fact which is useful for constructing generators (this is slightly non-trivial, and see \cite{ASSS16} for details and discussion). For this purpose, we use the following ``recipe'' from \cite{ASSS16} that gives a construction of such a map.

\begin{lemma}[\cite{ASSS16}]
\label{lem:faithful-recipe}
Let $\vec{F} = \set{F_1(\vecX), \ldots, F_m(\vecX)} \subseteq \F[\vecX]$ be a set of $N$-variate polynomials over $\F$ of degree at most $d$, such that $\trdeg{F_1, \ldots, F_m} \le r$, and suppose $\Char(\F) = 0$ or $\Char(\F) \ge d^r$. Let $C(y_1, \ldots, y_m) \in \F[\vecy]$ be any polynomial, and let $\Phi : \F[\vecX] \to \F[\vecz]$ be a homomorphism such that $\rank_{\F(\vecX)} (\Jac_{\vecX} (\vec{F})) = \rank_{\F(\vecz)}\Phi(\Jac_{\vecX}(\vec{F}))$. Consider the mapping $\Psi$ given by
\[
X_i \mapsto \left( \sum_{i=1}^r y_j t^{ij} \right) + \Phi(X_i).
\]
Then, it holds that $C(F_1, \ldots, F_m) \not\equiv 0$ if and only if $C(\Psi(F_1), \ldots, \Psi(F_m)) \not\equiv 0$.
 \end{lemma}

We now show how to construct a succinct generator for circuits of the form $C(F_1, \ldots, F_m)$ where $F_i$'s are polynomials of maximal sparsity $s$, and $\trdeg{F_1, \ldots, F_m} = k$.

\begin{lemma}
\label{lem:bms-small-transdeg}
Let $s, r, N \in \N$ and $m= r \log s + r \log r$.
Consider the polynomial map
\[
\cG^{\BMS}_{r,s} (y_1, \ldots, y_r, t_0,t_1,\ldots,t_n, w_1, \ldots, w_m, z_1, \ldots, z_m) := \Grc_{n,r} (\vecy, \vt) + \cG_{n,m}^{\SSSV} (\vecw, \vecz).
\]
Let $\vec{F} = \set{F_1(\vecX), \ldots, F_m(\vecX)} \subseteq \F[\vecX]$ be a set of $N$-variate polynomials over $\F$ of sparsity at most $s$, such that $\trdeg{F_1, \ldots, F_m} \le r$. Then for any polynomial of the form $G(\vecx) = C(F_1, \ldots, F_m)$, we have that $G \not\equiv 0$ if and only if $G \circ \cG^{\BMS}_{r,s} \not\equiv 0$.

Furthermore, the generator $\cG^{\BMS}_{r,s}$ is $\poly(r, \log s, n)$-$\SPS$ succinct.
\end{lemma}

\begin{proof}
By \autoref{lem:faithful-recipe} and \autoref{cl:succinct-vdm-is-vdm}, it is enough to show that the map $\cG_{n,m}^{\SSSV} (\vecw, \vecz)$ preserves the rank of the Jacobian matrix. This follows from the fact that each $r \times r$ minor of this matrix is a polynomial of sparsity at most $r! \cdot s^r$ (since taking derivatives can only decrease the sparsity), and from \autoref{cor:hitting-set-sparse}.

The succinctness claim follows from \autoref{cl:succinct-vdm-is-vdm} and \autoref{fact:sssv-succinct}.
\end{proof}

\begin{corollary}
There exists a $\poly(\log s, r, n)$-$\SPS$ succinct generator for the class of polynomials of the form $C(F_1, \ldots, F_m)$ such that each $F_i$ has sparsity at most $s$ and $\trdeg{F_1, \ldots, F_m} \le r$.
\end{corollary}

\section{Succinct Hitting Sets for Read-Once Oblivious Algebraic Branching Programs}\label{subsec:roabp}

In this section we construct a succinct hitting set for the class of read-once oblivious algebraic programs. Recall that in \autoref{subsec:comm-roabp} we have constructed a $\poly(\log w, \log n)$-$\SPS$ succinct generator for width-$w$ \emph{commutative} roABPs. For general ABPs, we are only able at this point to construct hitting sets that are width-$w^2$ roABP succinct: i.e., in the hitting set for width $w$ $N$-variate roABPs, each element is computed by a width $w^2$ $n$-variate roABP. Ideally, one would want to replace $w^2$ with $\polylog(w)$.

The definition of roABPs were given in \autoref{subsec:comm-roabp}. Throughout this section we assume that the ABP reads the variables in the order $X_1, X_2, \ldots , X_N$. In \autoref{subsec:roabp-order} we give some short remarks regarding different variable orderings.

Our construction is based on the following generator by Forbes and Shpilka \cite{FS13}.

\begin{lemma}[Forbes-Shpilka Generator for roABPs, Construction 3.13 in \cite{FS13}]
\label{lem:FS-generator}
Let $n \in \N$ and $N=2^n$. The following polynomial map $\cG : \F^{n+1} \to \F^N$ is a generator for width $w$, individual degree $d$, $N$-variate roABPs, in variable order $X_1, X_2, \ldots ,X_N$.

Let $\omega \in \F$ be of multiplicative order at least $(Ndw^2)^2$, and $\beta_1, \ldots, \beta_{w^2}$ be distinct elements of $\F$. Let $\setdef{p_\ell}{\ell \in [w^2]}$ be the Lagrange interpolation polynomials with respect to the $\beta_i$'s, i.e., $p_i (\beta_j) = 1$ if $i=j$ and $0$ otherwise.

Let $\cG : \F^{n+1} \to \F^N$ be the following polynomial map, whose output coordinates are indexed by vectors $\vecb \in \{0,1\}^n$.

\begin{equation}
\label{eq:FS-generator}
\cG^{\FS}_{\vecb} (\vecy) = \sum_{\ell_1, \ldots, \ell_{n} \in [w^2]} \prod_{i \in [n]}
\left(
(1-b_i) \cdot p_{\ell_{i-1}} (\omega^{\ell_i} y_i) + b_i \cdot p_{\ell_{i-1}} ( (\omega^{\ell_i} y_i)^{2^{i-1} dw^2} )
\right)
\cdot p_{\ell_{n}} (y_{n+1}),
\end{equation}
where we abuse notation by defining $p_{\ell_{0}} (t) = t$.
\end{lemma}

In \cite{FS13} (Lemma 3.18), it is shown that this map, for every fixed output coordinate $\vecb$, is computed by a width $w^2$ roABP in the variables $\vecy$. We, however, want to show that for every fixing $\vecy=\vecalpha$, there is a small roABP computing the polynomial whose coefficient vector is given by $(\cG_\vecb (\vecalpha))_{\vecb \in \{0,1\}^n}$. That is, for every choice of $\vecalpha$, and associating $\vecb$ with a subset of $[n]$, we want a polynomial in $x_1, \ldots, x_n$ such that the coefficient of $\vecx_b$ is $\cG_{\vecb}(\vecalpha)$.

\begin{definition}[Succinct Forbes-Shpilka Generator]
\label{def:succ-roabp}
Let $n, w \in \N$, and $\omega$, $p_i$'s as in \autoref{lem:FS-generator}. Define
\begin{align*}
P^{\FS} (x_1, \ldots, x_n, y_1, \ldots, y_{n+1}) &= 
\sum_{\ell_1, \ldots, \ell_{n} \in [w^2]} \prod_{i \in [n]}
\left(
p_{\ell_{i-1}} (\omega^{\ell_i} y_i) \right. \\ &+ \left. x_i \cdot p_{\ell_{i-1}} ( (\omega^{\ell_i} y_i)^{2^{i-1} dw^2} )
\right)
\cdot p_{\ell_{n}} (y_{n+1}). \qedhere
\end{align*}
\end{definition}

We first claim the the Forbes-Shpilka generator \eqref{eq:FS-generator} is given by the coefficient vector of this polynomial. 

\begin{claim}
\label{cl:succinct-FS-is-FS}
Assume the setup and notations of  \autoref{def:succ-roabp}. Then
$\coeff_{\vecx} (P^{\FS}) = \cG^{\FS}$.
\end{claim}

\begin{proof}
As explained earlier, we wish to show that the coefficient of $\vecx_b$ in the polynomial $P^{\FS}$ equals the $\vecb$-th coordinate of \eqref{eq:FS-generator}.

Fix a choice of $\ell_1, \ldots, \ell_n \in [w^2]$, and $\vecb \in \set{0,1}^n$. Consider the product
\[
\prod_{i \in [n]}
\left(
p_{\ell_{i-1}} (\omega^{\ell_i} y_i) + x_i \cdot p_{\ell_{i-1}} ( (\omega^{\ell_i} y_i)^{2^{i-1} dw^2} )
\right).
\]
Since the product is over distinct variables, there is exactly one way to obtain the monomial $\vecx_b = \prod_{i : b_i = 1} x_i$ in this product, and its coefficient will be
\begin{equation}
\label{eq:coeff-xb}
\prod_{i : b_i=1} p_{\ell_{i-1}} ( (\omega^{\ell_i} y_i)^{2^{i-1} dw^2} ) \; \cdot \; \prod_{i : b_i = 0}  p_{\ell_{i-1}} (\omega^{\ell_i} y_i) 
\end{equation}
Finally, observe that \eqref{eq:coeff-xb} exactly equals
\[
\prod_{i \in [n]}
\left(
(1-b_i) \cdot p_{\ell_{i-1}} (\omega^{\ell_i} y_i) + b_i \cdot p_{\ell_{i-1}} ( (\omega^{\ell_i} y_i)^{2^{i-1} dw^2} )
\right).
\]
This is true for every fixed choice of $\ell_1, \ldots, \ell_n$, and the claim now follows from the linearity of the coefficients map.
\end{proof}

We now show that for every fixing $\vecy=\vecalpha$, the polynomial $P^{\FS}(\vecx, \vecalpha)$ is computed by a small roABP.

\begin{claim}
\label{cl:succinct-FS-succinct}
For every setting $\vecy=\vecalpha$, the polynomial $P^{\FS}(\vecx, \vecalpha)$ in \autoref{def:succ-roabp} can be computed by a width $w^2$ roABP in variable order $x_1, x_2 \ldots, x_n$.
\end{claim}

\begin{proof}
The construction is straightforward from \autoref{def:succ-roabp}.
Layer $V_0$ contains the source vertex $s$ and layer $V_{n+1}$ the sink vertex $t$. Layers $V_1, \ldots, V_{n}$ each contain $w^2$ vertices labeled by the set $[w^2]$. For every $i \in [n]$ and every $\ell \in V_i$, there is an edge from each vertex in the previous layer, labeled by the linear function (in $x_i$)
\[
p_{\ell_{i-1}} (\omega^{\ell_i} \alpha_i) + x_i \cdot  p_{\ell_{i-1}} ( (\omega^{\ell_i} \alpha_i)^{2^{i-1} dw^2} ).
\]
Finally, all vertices in $V_n$ are connected to $t$ with an edge labeled $p_{\ell_{n}} (\alpha_{n+1})$.
\end{proof}

\begin{corollary}
\label{cor:succinct-roABP}
The Forbes-Shpilka generator given in \autoref{lem:FS-generator} is a width $w^2$-roABP succinct generator for degree $d$ roABPs that read the variables in order $X_1, X_2, \ldots , X_N$.
\end{corollary}

\begin{proof}
Immediate from \autoref{lem:FS-generator}, \autoref{cl:succinct-FS-is-FS} and \autoref{cl:succinct-FS-succinct}.
\end{proof}

\subsection{Different Variable Orderings}
\label{subsec:roabp-order}

The generator given by Forbes and Shpilka in \autoref{lem:FS-generator} hits roABPs that read the variables in the order $X_1, X_2, \ldots, X_N$ and not necessarily in any variable order. Obviously, we can apply a permutation $\sigma$ to the variables $x_1, \ldots, x_n$ in \autoref{def:succ-roabp} to obtain a roABP in the variables $\vecx$ in the order $\sigma$: the coefficient vector of this roABP hits roABPs in the variables $\vecX$ that read their variables in the order on $\set{X_1, \ldots, X_N}$ which is given by considering the lexicographic ordering induced on the set of multilinear monomials in $\set{x_1, \ldots, x_n}$ by the order $\sigma$, and using the canonical identification of a multilinear monomial with an index in $[N]$, say, using the binary representation. We call such an order relation on $[N]$ a \emph{monomial-compatible} ordering. Note that there are merely $n!$ such orderings among the $N!$ total orderings on $[N]$.

Since in our case we do not care about the \emph{size} of the hitting set, we can take the union of all $n!$ those succinct hitting sets to obtain the following corollary.

\begin{corollarywp}
There exists a width-$w^2$ roABP succinct hitting set for the class of width $w$, $N$ variate, and degree $d$ roABPs that read the variables in a monomial compatible ordering.
\end{corollarywp}

\section{Discussion and Open Problems}\label{sec:discussion}

In this work, we have shown that many of the hitting sets we know for restricted algebraic models of computation can be represented in a succinct form as coefficient vectors of small circuits. This gives some positive answers to \autoref{conj:complexity-barrier}, and points to the possibility of an algebraic natural proofs barrier. The main problem left open by this work is to construct succinct hitting sets for stronger models for which we know how to construct hitting sets efficiently.

For example, while we were able to construct a succinct generator for commutative roABPs, our construction for general roABPs is not fully succinct, and also works only in certain variable orderings. Despite several works that obtain quasi-polynomial size hitting sets for roABPs in any order (\cite{FSS14, agks15}), none of them seems to fit easily into the succinct setting, each for its own reasons.

For bounded-depth multilinear formulas, subexponential size hitting sets were obtained by Oliveira, Shpilka and Volk~\cite{osv16}. The construction there can be roughly described as hashing the $N$ variables into $N^{1-\varepsilon}$ buckets, and then hitting each bucket independently using a generator for roABPs (in fact, commutative roABPs will suffice). The main challenge here seems to be the hashing part, which (in the succinct setting) would involve hashing monomials, and ensuring that the coefficient vector that is obtained through this process has a small circuit for any possible hash function. 

The main technical tool which we do not know how to emulate in the succinct setting is the Klivans-Spielman~\cite{ks01} generator. In this generator, the variable $X_i$ is mapped to $t^{k^i \bmod p}$, where $t$ is a new indeterminate, $p$ is chosen from an appropriately large set of primes and $k$ from an appropriately large set of natural numbers. The main feature of this generator is that given a ``small'' enough set of monomials $\mathcal{M}$, the parameters $k$, $p$ can be chosen from a ``not too large'' set, such that all the monomials in $\mathcal{M}$ are given distinct weights, and this can be done in a black-box manner, that is, without knowing $\mathcal{M}$, but only an upper bound on its size. Indeed, the noticeable difference from the constructions we have given in this paper is the  \emph{exponential} dependence on $i$ in the exponent of $t$, a feature which is not clear how to emulate in the succinct setting.

The main application of the Klivans and Spielman construction is to construct hitting sets for sparse polynomials.  While we are unable to make the resulting hitting set succinct, we developed an alternate hitting set which we succeeded in making succinct.  However, the Klivans and Spielman construction (or otherwise similar ideas) has also found applications beyond the class of sparse polynomials, such as in the construction hitting sets for roABPs in unknown order from the work of Agrawal, Gurjar, Korwar and Saxena~\cite{agks15}. Unfortunately, such works seem to rely heavily on properties of the Klivans and Spielman construction beyond that of just hitting sparse polynomials, and as such we are currently unable to make these hitting sets succinct.

A particular interesting application of the Klivans and Spielman construction is in the recent works of Fenner, Gurjar and Thierauf~\cite{FGT16} and its generalization by Gurjar and Thierauf \cite{GT17}.  These works construct hitting sets for the class of determinants of ``read-once matrices'', which are polynomials of the form $\det M$, where $M$ is a matrix in which each entry contains a variable $x_{i,j}$ or a field constant, and each variable appears at most once in the matrix. While this class of polynomials is very restricted, the partial derivative matrix used by Nisan~\cite{nis91}, Raz~\cite{raz2004}, and Raz-Yehudayoff~\cite{raz-yehudayoff}, is a read-once matrix. As such, the lower bounds proved in these papers are algebraically natural and the distinguisher used is a read-once determinant. The work of Raz and Yehudayoff~\cite{raz-yehudayoff} in particular shows that a read-once determinant can vanish on the coefficient vectors of constant-depth multilinear formulas, and as most of the constructions in this paper have this form this shows that these constructions cannot be succinct hitting sets for read-once determinants, and hence new ideas are needed.  Indeed, if one could establish a circuit class $\cC$ where there are $\cC$-succinct hitting sets for read-once determinants then this would show that no proof technique following the ideas of the above works can prove lower bounds for the class $\cC$. Such a result would be very interesting as those lower bounds methods are still very much state-of-the-art.

As mentioned earlier, stronger evidence towards an algebraic natural proofs barrier can also be obtained by designing pseudorandom polynomials whose security is based on widely-believed cryptographic assumptions. In particular, one possible approach is obtaining evidence in favor of the determinant-based construction of Aaronson and Drucker~\cite{aaronson-blog-post}.

\section*{Acknowledgements}
We thank Scott Aaronson, Andy Drucker, Josh Grochow, Mrinal Kumar, Shubhangi Saraf and Dor Minzer for useful conversations regarding this work. We also thank the anonymous reviewers for their careful reading of this paper and for many useful comments.

\bibliographystyle{customurlbst/alphaurlpp}
\bibliography{references}

\end{document}